\documentclass{scrartcl}
\usepackage[numbers]{natbib}
\usepackage{fontenc}
\usepackage{shadethm}
\usepackage{amsthm}
\usepackage{float}
\usepackage{bm}
\usepackage{multicol,lipsum}
\usepackage{mathtools}
\usepackage{graphicx}
\usepackage{xcolor}
\usepackage{subfig}
\usepackage{stmaryrd}
\usepackage{mathrsfs}
\usepackage{amsmath}
\usepackage{amssymb}
\usepackage{wasysym}
\usepackage{sectsty}
\usepackage[hyperfootnotes=false]{hyperref}
\usepackage[a4paper, left=2.7cm, right=2.7cm, top=2.5cm]{geometry}
\usepackage{chngcntr}
\counterwithin*{equation}{section}

\usepackage{cleveref}
\DeclareMathAlphabet{\mathpzc}{OT1}{pzc}{m}{it}

\sectionfont{\normalsize\centering}
\subsectionfont{\footnotesize\centering}

\theoremstyle{plain}
\newtheorem{thm}{Theorem}[section] 

\theoremstyle{definition}
\newtheorem{defn}[thm]{Definition} 
\newtheorem{lem}[thm]{Lemma}
\newtheorem{prop}[thm]{Proposition}
\newtheorem{rem}[thm]{Remark}
\newtheorem{cor}[thm]{Corollary}

\def\XXint#1#2#3{{\setbox0=\hbox{$#1{#2#3}{\int}$ }
		\vcenter{\hbox{$#2#3$ }}\kern-.6\wd0}}

\usepackage[utf8]{inputenc}
\usepackage[german,english,russian]{babel}

\newcounter{MPequ}

\newcounter{AppA}
\newenvironment{AppA}
{\stepcounter{AppA}%
	\addtocounter{equation}{0}%
	\equation}
{\endequation}
\newcounter{AppB}
\newenvironment{AppB}
{\stepcounter{AppB}%
	\addtocounter{equation}{0}%
	\equation}
{\endequation}
\newcounter{AppC}
\newenvironment{AppC}
{\stepcounter{AppC}%
	\addtocounter{equation}{0}%
	\equation}
{\endequation}
\newcounter{AppD}

\newcounter{AppE}

\begin{document}\selectlanguage{english}
\begin{center}
\normalsize \textbf{\textsf{An analytic formula for surface currents generating prescribed plasma equilibrium fields}}
\end{center}
\begin{center}
 Wadim Gerner\footnote{\textit{E-mail address:} \href{mailto:wadim.gerner@edu.unige.it}{wadim.gerner@edu.unige.it}}
\end{center}
\begin{center}
{\footnotesize MaLGa Center, Department of Mathematics, Department of Excellence 2023-2027, University of Genoa, Via Dodecaneso 35, 16146 Genova, Italy}
\end{center}
{\small \textbf{Abstract:} 
Given a plasma domain $P\subset\mathbb{R}^3$, a plasma equilibrium field $B$ on $P$ and a coil winding surface $\Sigma$ surrounding $P$, we provide an analytic formula whose output is a surface current distribution $j$ on $\Sigma$ such that $\operatorname{BS}(j)+\operatorname{BS}(\operatorname{curl}(B))=B$ in $P$, i.e. the combination of the plasma current magnetic field and the surface current magnetic field exactly produce the full plasma equilibrium field. We use this new formula to provide a theoretical explanation of the empirical phenomenon that currents tend to attain their maximal strength at points which correspond to points of highest normal magnetic curvature on the plasma boundary. Some discussions regarding aspects of numerical approximations are also included.
\newline
\newline
{\small \textit{Keywords}: Coil winding surface model, Stellarator, Magnetic field approximation}
\begin{multicols}{2}
\section{Introduction}
\label{Intro}
In stellarator devices, cf. \cite{Spi58}, the hot fusion plasma is confined within a plasma region by means of a magnetic field which is generated by complex coil configurations. Traditionally the plasma region, plasma equilibrium field and coil configurations are obtained by means of a two step procedure:
\begin{enumerate}
	\item First, a plasma equilibrium field is sought together with a plasma region supporting it. This is achieved by utilising plasma equilibrium codes such as VMEC or DESC, cf. \cite{HirshmanWhitson83},\cite{DK20}. During this step the plasma is optimised for desirable confinement properties.
	\item In a second step one seeks coil configurations which are capable to support the plasma equilibrium field that was found in the first step.
\end{enumerate}
To be more precise, a plasma equilibrium field is a magnetic field $B$ which satisfies 
\begin{gather}
	\nonumber
	B\times \operatorname{curl}(B)=\nabla p\text{, }\operatorname{div}(B)=0\text{ in }P
	\\
	\label{S1EExtra}
	\text{ and }B,\operatorname{curl}(B)\parallel \partial P
\end{gather}
where $P\subset \mathbb{R}^3$ is the plasma region (and naturally assumed throughout to be topologically a solid torus). According to Amp\'{e}re's law, the plasma current $J$ is given by $J=\operatorname{curl}(B)$. The induced magnetic field is then given by the Biot-Savart law as
\begin{gather}
	\nonumber
	\operatorname{BS}_P(J)(x)=\frac{1}{4\pi}\int_{P}J(y)\times \frac{x-y}{|x-y|^3}d^3y.
\end{gather}
One can then define the magnetic field $$B_T:=B-\operatorname{BS}_P(J)$$ and it induces a vacuum field, i.e. $$\operatorname{div}(B_T)=0=\operatorname{curl}(B_T).$$ One then seeks to find a coil configuration and corresponding electric current densities supported on these coils which reproduce the desired target field $B_T$, so that the combination of the coil induced field and plasma induced field produce together the desired equilibrium field $B$.

In the following we work entirely in the coil winding surface model, cf. \cite{M87}, which assumes that the coil current $j$ is supported on a toroidal surface $T^2\cong \Sigma\subset\mathbb{R}^3$ surrounding the plasma region at a positive distance. In this context one is therefore looking for a suitable $\Sigma$ and surface current distribution $j$ on $\Sigma$ such that
\begin{gather}
\nonumber
\operatorname{BS}_{\Sigma}(j)(x)=\frac{1}{4\pi}\int_{\Sigma}j(y)\times \frac{x-y}{|x-y|^3}dS(y)
\\
\nonumber
=B_T(x)\text{ for all }x\in P.	
\end{gather}
In general, if $B_T$ is an arbitrary prescribed vacuum field on $P$, i.e. satisfying $\operatorname{div}(B_T)=0=\operatorname{curl}(B_T)$, then there does not exist such a current distribution $j$. From a theoretical point of view, \cite[Theorem 1.1]{G24}, one can at least approximate any such target field to arbitrary precision. More specifically this means that for a given coil winding surface (CWS) $\Sigma$ and every $\epsilon>0$ there is a current $j_{\epsilon}$ supported on $\Sigma$ with $$\|\operatorname{BS}_{\Sigma}(j_{\epsilon})-B_T\|_{L^2(P)}\leq \epsilon.$$
From a practical point of view, in order to obtain such current distributions, one may employ a modified REGCOIL-procedure, cf. \cite{L17} for the original REGCOIL reference. The original REGCOIL procedure proceeds as follows: Given a CWS $\Sigma$ and regularising parameter $\lambda>0$ one considers the minimisation problem
\begin{gather}
	\nonumber
	\|\mathcal{N}\cdot \operatorname{BS}_{\Sigma}(j)-\mathcal{N}\cdot B_T\|^2_{L^2(\partial P)}+\lambda\|j\|^2_{L^2(\Sigma)}
	\\
	\nonumber
	\rightarrowtail \operatorname{minmise}
\end{gather}
which admit unique minimisers. One then uses the magnetic fields induced by these minimising currents to approximate $B_T$. These approximations improve with decreasing value of $\lambda$ but usually at the cost of increased current field line complexity. If one instead considers the following modified minimisation problem
\begin{gather}
	\label{S1E1}
	\|\operatorname{BS}_{\Sigma}(j)-B_T\|^2_{L^2(P)}+\lambda \|j\|^2_{L^2(\Sigma)}\rightarrowtail \operatorname{minimise}
\end{gather}
then these modified problems also admit unique minimisers which we denote by $j_{\lambda}$. It follows then from \cite[Corollary 4.3]{G24} that $$\|\operatorname{BS}_{\Sigma}(j_{\lambda})-B_T\|_{L^2(P)}\rightarrow 0\text{ as }\lambda \searrow 0.$$ Despite this theoretical result, the corresponding current field lines of the $j_{\lambda}$ are generally becoming more complex as one decreases $\lambda$. To address the coil complexity one may attempt to identify features of the original plasma equilibrium field $B$ which influence coil complexity\slash current field line complexity, cf. \cite{KLM24}. Since there is no known explicit formula for the current distribution $j$ which exactly reproduces the given target field on $P$, these findings are usually empiric and work with the approximating current $j_{\lambda}$ (or with minimisers of related objective functions, e.g. the original REGCOIL functional).

In a recent letter, cf. \cite{RodSengArXiv}, the special situation was considered in which the plasma equilibrium field $B$ was itself assumed to be a vacuum field and the coil winding surface has been assumed to coincide with the plasma boundary. In that case there is a known formula for the current distribution $j$ on $\Sigma=\partial P$ which exactly reproduces $B$ according to the virtual-casing principle \cite{Hanson15}. More precisely, under these special circumstances, $\operatorname{curl}(B)=0=\operatorname{div}(B)$, $B\parallel \partial P$, we have the identity
\begin{gather}
	\label{S1E2}
	\operatorname{BS}_{\partial P}(B\times\mathcal{N})=B
\end{gather} 
where $\mathcal{N}$ is the outward unit normal to $\partial P$. We note that since $\operatorname{curl}(B)=0$, $j:=B\times \mathcal{N}$ is div-free and thus indeed a valid electric current distribution. It has then been discussed in \cite{RodSengArXiv} that certain properties of $B$ influence the complexity of the field lines of $j$ and that these properties of $B$ appear to remain of relevance in more practical situations where the CWS has a positive distance to the plasma boundary.
\newline
\newline
The goal of the present work is to provide for a given CWS $\Sigma$ (of positive distance to the plasma) an explicit analytic formula for a current distribution $j$ on $\Sigma$ which precisely supports a given plasma equilibrium field $B$ (which does not need to be a vacuum field) on a plasma region $P$. In addition, the constructed current distribution also supports a compatible vacuum field within the region between the plasma boundary and the CWS.

The remainder of the manuscript is structured as follows: In \Cref{Section2} we recall some necessary mathematical notions before we clearly state the analytic formula. In \Cref{Section3} we discuss some numerical aspects of our analytic formula. In particular, we discuss some approaches which can be used to approximate the analytic quantities in question. In \Cref{Section4} we discuss some properties of the dynamics of the resulting currents. In \Cref{Section5} we discuss implications of our findings and compare our results with previously known results in the literature. In \Cref{Section6} we summarise the main findings of the paper. The appendix contains mathematically rigorous proofs of the analytic formula and some additional results which are presented throughout the main body of this manuscript.
\section{An analytic formula for the current distribution}
\label{Section2}
In the next subsections we introduce necessary notions which are required to state the analytic formula.
\subsection{Compatible vacuum fields}
Suppose that $B$ is a plasma equilibrium field, i.e. it satisfies (\ref{S1EExtra}) on some plasma domain $P$. For a given CWS $\Sigma$ we then require the existence of a vacuum field $V$ within the region between $\Sigma$ and $\partial P$, i.e. $\operatorname{curl}(V)=0=\operatorname{div}(V)$, which is compatible with the plasma equilibrium field $B$, i.e. we require $V|_{\partial P}=B|_{\partial P}$. If such a vacuum field does not exist, then $B$ is incompatible with the existence of a corresponding vacuum equilibrium field and thus incompatible with a steady operation of the designed stellarator. It is therefore of interest to understand under what conditions such vacuum fields exist. We state here a known mathematical fact, \cite[Theorem 3.1]{EP12},\cite[Theorem D.1]{Ger26CoilGeometry}
\begin{thm}[{\cite[Theorem 3.1]{EP12}}, Existence of compatible vacuum fields]
	\label{S2T1}
	Let $B$ be a plasma equilibrium field, i.e. it solves (\ref{S1EExtra}), within a plasma region $P$. Then there exists an open subset $U\subset\mathbb{R}^3$ such that $\overline{P}\subset U$ and there exists a unique vacuum field $V$, i.e. $\operatorname{div}(V)=0=\operatorname{curl}(V)$, on $U\setminus P$ satisfying $V|_{\partial P}=B|_{\partial P}$.
\end{thm}
We emphasise that while every plasma equilibrium field supports a compatible vacuum field, the size of $U$ may depend on $B|_{\partial P}$. In addition, we remark that if the pressure is constant along $\partial P$ (as is customary in stellarator designs), then the vector field $\widetilde{B}:=B$ on $P$ and $\widetilde{B}:=V$ on $U\setminus P$ provides a weak solution to the plasma equilibrium equations, i.e. it solves $\widetilde{B}\times \operatorname{curl}(\widetilde{B})=\nabla \widetilde{p}$ where $\widetilde{p}:=p$ in $P$ and $\widetilde{p}:=p|_{\partial P}$ on $U\setminus P$. To simplify notation we write in the following $B$ instead of $\widetilde{B}$ to refer to the extended plasma equilibrium field. Notice that one should aim to select a CWS $\Sigma$ such that $\Sigma\subset U\setminus \overline{P}$ because that ensures that there exists a suitable vacuum region between $\Sigma$ and $P$ which is compatible with our plasma equilibrium field $B$.
\subsection{The double layer potential}
Given a bounded smooth domain $\Omega\subset \mathbb{R}^3$, the transpose of the double layer potential, denoted $w^{\operatorname{Tr}}_{\Omega}$ is the following operator which takes as input a function $\phi$ on the boundary of $\partial \Omega$ and outputs a function $w^{\operatorname{Tr}}_{\Omega}(\phi)$ according to the following formula
\begin{gather}
	\label{S2E1}
	w^{\operatorname{Tr}}_{\Omega}(\phi)(x):=\frac{\int_{\partial\Omega}\phi(y)\mathcal{N}(x)\cdot \frac{x-y}{|x-y|^3}dS(y)}{4\pi}
\end{gather}
where $\mathcal{N}$ denotes the outward unit normal on $\partial\Omega=\Sigma$. To state our analytic formula we will need to invoke the following result.
\begin{thm}[{\cite[Lemma 4.2]{Gerner26KernelImage}}]
	\label{S2T2}
	The operator $\frac{\operatorname{Id}}{2}+w^{\operatorname{Tr}}_{\Omega}$ is invertible.
\end{thm}
To avoid technicalities we do not specify the precise domain and range of the operator and instead refer the mathematically interested reader to \cite[Lemma 4.2]{Gerner26KernelImage} for the details, see also \Cref{AppB}. We emphasise that $w^{\operatorname{Tr}}_{\Omega}$ is linear and thus so is the inverse of $\frac{\operatorname{Id}}{2}+w^{\operatorname{Tr}}_{\Omega}$.
\subsection{Harmonic Neumann fields and the kernel of the Biot-Savart operator}
Given a CWS $\Sigma$ it is well-known, \cite{Li88}, that $\mathbb{R}^3\setminus \Sigma$ consists of one bounded component and one unbounded component. In the following we always denote by $\Omega$ the bounded component and by $\widetilde{\Omega}$ the unbounded component. In particular, $\partial \Omega=\Sigma=\partial\widetilde{\Omega}$. Before we come to a characterisation of the kernel of the Biot-Savart operator we make the following definition.
\begin{defn}[Harmonic Neumann fields on exterior domains]
	\label{S2D3}
	Let $\Sigma$ be a CWS and let $\widetilde{\Omega}\subset \mathbb{R}^3$ be the unbounded component of $\mathbb{R}^3\setminus \Sigma$. We define the space of harmonic Neumann fields, denoted by $\mathcal{H}_N(\widetilde{\Omega})$, as follows
	\begin{gather}
		\nonumber
		\mathcal{H}_N(\widetilde{\Omega}):=\left\{\Gamma \in H^1(\widetilde{\Omega},\mathbb{R}^3)\mid \operatorname{curl}(\Gamma)=0=\operatorname{div}(\Gamma)\text{, }
		\right.
		\\
		\label{S2E2}
		\left. \Gamma\parallel \partial \widetilde{\Omega}\text{, }|\Gamma(x)|\rightarrow 0\text{ as }|x|\rightarrow \infty\right\}.
	\end{gather}
\end{defn}
The space $\mathcal{H}_N(\widetilde{\Omega})$ is a vector space and its dimension is related to the topology of $\Sigma$. We have the following result, see \Cref{AppA} for a proof.
\begin{prop}[Dimension of the space of harmonic Neumann fields]
	\label{S2P4}
	Let $\Sigma\subset\mathbb{R}^3$ be a (toroidal) CWS, then $\operatorname{dim}\left(\mathcal{H}_N(\widetilde{\Omega})\right)=1$ and for any fixed (purely) poloidal loop $\sigma_p\subset \Sigma$ there exists a unique element $\widetilde{\Gamma}_p\in \mathcal{H}_N(\widetilde{\Omega})$ satisfying $\int_{\sigma_p}\widetilde{\Gamma}_p=1$.
\end{prop}
The relevance of the space of harmonic Neumann fields in the context of coil design is the following result, cf. \Cref{AppA} for a proof.
\begin{thm}[Characterisation of the kernel of the Biot-Savart operator]
	\label{S2T5}
	Let $\Sigma\subset\mathbb{R}^3$ be a (toroidal) CWS, then the vector space $$\operatorname{Ker}(\operatorname{BS}_{\Sigma}):=\left\{j\mid \operatorname{BS}_{\Sigma}(j)=0\text{ in }\Omega\right\}$$
	is one-dimensional and is spanned by $j_0:=\mathcal{N}\times \widetilde{\Gamma}_p$ where $\widetilde{\Gamma}_p$ is as in \Cref{S2P4}.
\end{thm}
\subsection{The analytic formula}
Before we state our formula we require one final ingredient.
\begin{lem}
	\label{S2L6}
	Let $B$ be a plasma equilibrium field on a plasma domain $P$ and let $U$ be a neighbourhood of $\overline{P}$ on which $B$ admits a compatible vacuum extension, cf. \Cref{S2T1}. Then for every CWS $\Sigma$ with $\Sigma\subset U\setminus \overline{P}$ there exists a solution $f$ to the following boundary value problem (BVP)
	\begin{gather}
		\nonumber
		\Delta f=0\text{ in }\Omega\text{, }
		\\
		\label{S2E3}
		\mathcal{N}\cdot \nabla f=\left(\left(\frac{\operatorname{Id}}{2}+w^{\operatorname{Tr}}_{\Omega}\right)^{-1}-\operatorname{Id}\right)(\mathcal{N}\cdot B)
	\end{gather}
	where $\Omega$ denotes the finite region enclosed by $\Sigma$.
\end{lem}
The main finding of the present manuscript is the following result, see \Cref{AppB} for a proof.
\begin{thm}[Analytic current formula]
	\label{S2T7}
	Let $B$ be a plasma equilibrium field, i.e. it satisfies (\ref{S1EExtra}), on some plasma region $P$. Let $U$ be an open set on which $B$ admits a compatible vacuum extension, cf. \Cref{S2T1}, and let $\Sigma \subset U\setminus \overline{P}$ be any fixed CWS. Denote by $f$ any solution to the BVP (\ref{S2E3}), then for every $\alpha\in \mathbb{R}$ we have
	\begin{gather}
		\nonumber
		\operatorname{BS}_{\Sigma}(B\times \mathcal{N}+\nabla f\times \mathcal{N}-\alpha \widetilde{\Gamma}_p\times \mathcal{N})
		\\
		\nonumber
		=B-\operatorname{BS}_{P}(J)\text{ in }\Omega
	\end{gather}
	where $J=\operatorname{curl}(B)$ is the plasma current and $\widetilde{\Gamma}_p$ is as in \Cref{S2P4}. In addition, the field lines of the corresponding current are as poloidal as possible for the choice $\alpha=\int_{\sigma_p}B$.
\end{thm}
We point out that if we define $$j_{\alpha}:=B\times \mathcal{N}+\nabla f\times \mathcal{N}-\alpha \widetilde{\Gamma}_p\times \mathcal{N},$$ then this defines a div-free vector field on $\Sigma$ and thus the $j_{\alpha}$ are admissible current distributions on $\Sigma$. Further, the term $\alpha \widetilde{\Gamma}_p\times \mathcal{N}$ does not affect the magnetic field within $\Omega$ and is thus a degree of freedom. On average, the field lines of $\widetilde{\Gamma}_p\times \mathcal{N}$ do not wind around poloidally around $\Sigma$, but instead, on average, wind around solely toroidally around $\Sigma$, cf. \cite[Theorem B.1]{Gerner26KernelImage}. Therefore, the term $\alpha \widetilde{\Gamma}_p\times \mathcal{N}$ allows us to control the toroidal complexity of the current distribution.
\section{Numerical aspects}
\label{Section3}
In order to compute the relevant quantity
\begin{gather}
	\nonumber
	j_{\alpha}=B\times \mathcal{N}+\nabla f\times \mathcal{N}-\alpha\widetilde{\Gamma}_p\times \mathcal{N}
\end{gather}
one needs to compute the quantities
\begin{enumerate}
	\item $B|_{\Sigma}$ or more generally the vacuum field compatible with $B$ within some open neighbourhood of $\overline{P}$,
	\item a solution $f$ to the BVP (\ref{S2E3}),
	\item the harmonic Neumann field $\widetilde{\Gamma}_p$.
\end{enumerate}
In the following subsections we discuss how each of these quantities may be approximated.
\subsection{Computing the unique vacuum field}
To get the main idea across we consider here the situation of a flat surface. Let $$H:=\{(x,y,z)\in \mathbb{R}^3\mid z\geq 0\}$$
be the upper half space and suppose that
\begin{gather}
	\nonumber
	\operatorname{curl}(\Gamma)=0=\operatorname{div}(\Gamma)\text{ in }H,
	\\
	\nonumber
	\Gamma^3(x,y,0)=0\text{ for all }(x,y)\in \mathbb{R}^2
\end{gather}
where $\Gamma^k$ denotes the $k$-th component of $\Gamma$ in Euclidean coordinates. The goal is to obtain $\Gamma(x,y,z)$ from knowledge of $\Gamma|_{\partial H}$, i.e. $\Gamma(x,y,0)$. To this end one may perform a power series expansion around $z=0$ and write
\begin{gather}
	\nonumber
	\Gamma(x,y,z)=\sum_{k=0}^{\infty}(\partial^{(k)}_z\Gamma)(x,y,0)\frac{z^k}{k!}.
\end{gather}
We are therefore left with expressing $(\partial^{(k)}_z\Gamma)(x,y,0)$ in terms of $\Gamma(x,y,0)$ and its derivatives. We know that
\begin{gather}
	\nonumber
	\partial_z\Gamma^3=-(\partial_x\Gamma^1)-(\partial_y\Gamma^2)\text{, }
	\\
	\nonumber
	\partial_z\Gamma^2=\partial_y\Gamma^3\text{, }\partial_z\Gamma^1=\partial_x\Gamma^3.
\end{gather}
This shows that we can express $(\partial_z\Gamma)(x,y,0)$ in terms of $(\partial_x\Gamma)(x,y,0),(\partial_y\Gamma)(x,y,0)$ and so $(\partial_z\Gamma)(x,y,0)$ is determined by the behaviour of $\Gamma(x,y,0)$ alone. Higher order derivatives can be computed accordingly. For example,
\begin{gather}
	\nonumber
	\partial^2_{zz}\Gamma^3=-(\partial_z\partial_x\Gamma^1)-(\partial_z\partial_y\Gamma^2)
	\\
	\nonumber
	=-\partial^2_{xx}\Gamma^3-\partial^2_{yy}\Gamma^3
\end{gather}
and similarly all higher derivatives in $z$-direction can be expressed through derivatives in $x$ and $y$ direction alone. Hence the vacuum field is uniquely determined by its values on the boundary. In practice, the CWS $\Sigma$ will be curved and one has to take into account geometric properties of the CWS when computing the vacuum field in terms of a power series expansion. We refer the reader to \cite[Appendix D]{Ger26CoilGeometry} for a detailed account of this approach in the context of curved CWSs. We just point out that the main idea is to work in tubular neighbourhood coordinates to perform the expansion.

In contrast to that, there has also been another approach presented in a different context in \cite{LiuZZ25}, where the computation of a vacuum field around a flux surface in Boozer coordinates was presented. We point out that the approach in \cite{LiuZZ25} assumes that the vacuum field admits a flux function (and thus a first integral). However, in general it might be that the vacuum field corresponding to a certain plasma equilibrium might not admit first integrals and as such the approach in \cite{LiuZZ25} cannot always be used to compute the compatible vacuum field. Other vacuum field computations of relevance in the plasma fusion research have also been considered in \cite{Malhotra2020}. However, \cite{Malhotra2020} does not consider the computation of the compatible vacuum field directly.
\subsection{Solving the BVP}
In order to compute a solution to the BVP (\ref{S2E3}) one needs to specify the average of $f$ to obtain unique solutions, so that one may assume $\int_{\Omega}fdx=0$ when looking for a solution. The main difficulty in finding a solution to (\ref{S2E3}) consists in computing $\left(\frac{\operatorname{Id}}{2}+w^{\operatorname{Tr}}_{\Omega}\right)^{-1}$. To achieve this one may exploit that this inverse admits a Neumann series expansion, cf. \cite[Formula after equation (4.15)]{Gerner26KernelImage}, $$ \left(\frac{\operatorname{Id}}{2}+w^{\operatorname{Tr}}_{\Omega}\right)^{-1}=\sum_{k=0}^{\infty}\left(\frac{\operatorname{Id}}{2}-w^{\operatorname{Tr}}_{\Omega}\right)^k$$ and thus
$$\left(\frac{\operatorname{Id}}{2}+w^{\operatorname{Tr}}_{\Omega}\right)^{-1}-\operatorname{Id}=\sum_{k=1}^{\infty}\left(\frac{\operatorname{Id}}{2}-w^{\operatorname{Tr}}_{\Omega}\right)^k.$$
The main observation now is that this series converges exponentially fast in the sense that there exists some $0<\lambda<1$ (depending on the CWS $\Sigma$) such that for every natural number $n$ we have
\begin{gather}
	\nonumber
	\left\|\sum_{k=1}^{\infty}\left(\frac{\operatorname{Id}}{2}-w^{\operatorname{Tr}}_{\Omega}\right)^k-\sum_{k=1}^{n}\left(\frac{\operatorname{Id}}{2}-w^{\operatorname{Tr}}_{\Omega}\right)^k\right\|\leq \frac{\lambda^{n+1}}{1-\lambda}
\end{gather}
for a suitable norm $\|\cdot \|$, cf. \cite[Corollary 4.4]{Gerner26KernelImage} for the details. Hence, in practice, it should suffice to take the first few terms of this expansion in order to efficiently compute the boundary values. The most difficult task which is left is therefore to find approximations for $w_{\Omega}^{\operatorname{Tr}}(\phi)$ for a given $\phi$. To ease the computation of $w^{\operatorname{Tr}}_{\Omega}$ one may observe that
\begin{gather}
	\nonumber
	\int_{\partial\Omega}w^{\operatorname{Tr}}_{\Omega}(\phi)(x)\cdot \psi(x)dS=\int_{\partial\Omega}w_{\Omega}(\psi)(x)\cdot \phi(x)dS
\end{gather}
where $w_{\Omega}(\psi)$, the double layer potential of $\psi$, is defined by $$w_{\Omega}(\psi)(x):=\frac{\int_{\partial\Omega}\psi(y)\mathcal{N}(y)\cdot \frac{y-x}{|y-x|^3}dS(y)}{4\pi}.$$ An efficient way to compute the double layer potential is presented in \cite{Malhotra2020} so that overall it is also possible to numerically compute a solution to the BVP (\ref{S2E3}).
\subsection{Computing the harmonic Neumann field}
The main issue regarding the numerical computation of the harmonic Neumann field $\widetilde{\Gamma}_p$ is that it is defined on an unbounded domain. We notice however, that only $\widetilde{\Gamma}_p\times \mathcal{N}$ enters the analytic formula for the current. It is hence enough to approximate $\widetilde{\Gamma}_p$ on some neighbourhood around the CWS. Before we discuss how to achieve this we introduce the following space of harmonic Neumann fields on a bounded domain $U\subset \mathbb{R}^3$
\begin{gather}
	\nonumber
	\mathcal{H}_N(U):=\{\Gamma \mid \operatorname{curl}(\Gamma)=0=\operatorname{div}(\Gamma)\text{, }\Gamma \parallel \partial U\}.
\end{gather}
For a given CWS $\Sigma$, fix some $0<r$ so large that $\overline{\Omega}\subset B_r(0)$ where $\Omega$ denotes the finite region enclosed by $\Sigma$ and where $B_r(0)$ denotes the Euclidean ball of radius $r$ centred at $0$. Then for every such $r$ we have $\operatorname{dim}\left(\mathcal{H}_N(B_r(0)\setminus \overline{\Omega})\right)=1$ and for a fixed poloidal loop $\sigma_p$ on $\partial\Omega=\Sigma$ there is a unique element $\Gamma_r\in \mathcal{H}_N(B_r(0)\setminus \overline{\Omega})$ with $\int_{\sigma_p}\Gamma_r=1$. We then have the following result, see \Cref{AppA} for a proof.
\begin{prop}[Approximating the exterior Harmonic Neumann field]
	\label{S3P1}
	Let $\Sigma$ be our (toroidal) CWS and let $\Omega$ be the finite region enclosed by $\Sigma$. Then with the notation above we have $$\Gamma_r\times \mathcal{N}\rightarrow \widetilde{\Gamma}_p\times \mathcal{N} $$
	as $r\rightarrow \infty$. Even more, we have a corresponding estimate $$\|\Gamma_r-\widetilde{\Gamma}_p\|_{L^2(B_{r}(0)\setminus \overline{\Omega})}\in \mathcal{O}\left(r^{-\frac{3}{2}}\right)\text{ as }r\rightarrow\infty.$$
\end{prop}
\begin{figure}[H]
	\centering
	\includegraphics[width=0.405\textwidth]{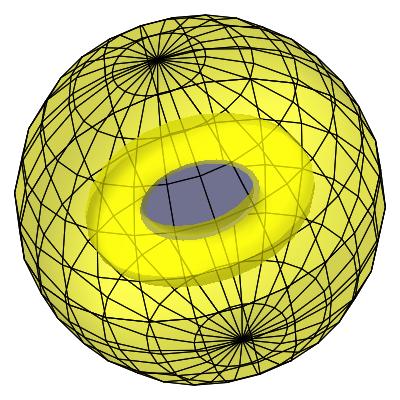}
	\caption{
		The toroidal CWS $\Sigma$ is seen inside the ball containing it. A cut surface is shown in blue; a cut surface is any surface $S$ inside $U_r:=B_r(0)\setminus \overline{\Omega}$ with the property that $U_r\setminus S$ is simply-connected. One can then look for a solution $\widetilde{f}_r$ of the BVP $\Delta \widetilde{f}_r=0$ in $U_r$, $\mathcal{N}\cdot \nabla \widetilde{f}_r=0$ along the yellow surfaces and $\widetilde{f}^{+}_r|_{S}=+1$, $\widetilde{f}^{-}_r|_{S}=0$ where $\widetilde{f}^{\pm}_r$ denotes the value of $\widetilde{f}_r$ when it approaches the blue surface from the top and bottom respectively. One then has the identity $\Gamma_r=\nabla \widetilde{f}_r$.} 
	\label{FigureBVP}
\end{figure}
\Cref{S3P1} tells us that it is enough to compute the harmonic Neumann fields of domains of the form $B_r(0)\setminus \overline{\Omega}$. This can be achieved by introducing a cut surface and solving a PDE for a scalar function with a jump condition along the cut surface, see \Cref{FigureBVP}.
\section{Current dynamics}
\label{Section4}
Since it is natural to seek current distributions whose field lines are as poloidal as possible, we consider here solely the situation in which $\alpha=\int_{\sigma_p}B$, cf. \Cref{S2T7}. We then introduce the \textit{effective magnetic field} $B_{\operatorname{eff}}$ on a given CWS $\Sigma$ as follows
\begin{gather}
	\label{S4E1}
	B_{\operatorname{eff}}:=B-\left(\int_{\sigma_p}B\right)\widetilde{\Gamma}_p+\nabla f
\end{gather}
so that in turn the unique, most poloidal, current distribution $j_p$ may be expressed compactly as
\begin{gather}
	\label{S4E2}
	j_p=B_{\operatorname{eff}}\times \mathcal{N}.
\end{gather}
We observe that (\ref{S4E2}) implies that the current $j_p$ only sees the part of $B_{\operatorname{eff}}$ which is tangent to the CWS $\Sigma$. In particular, if we want to be able to relate dynamical properties of $B_{\operatorname{eff}}$ (and thus those of $B$) to the dynamical properties of $j_p$ one should aim to select the CWS $\Sigma$ in such a way that $B_{\operatorname{eff}}\parallel \Sigma$. Otherwise it may happen that even though $B_{\operatorname{eff}}$ has well-behaved dynamics, there is a point $x$ on the CWS such that $B_{\operatorname{eff}}(x)\perp \Sigma$ and thus $j_p(x)=0$. It is however known that the presence of zeros of an electric current distribution lead to the presence of saddle and centre regions which create complicated current field line dynamics, cf. \cite[Theorem 1.1]{GerNB26ElectricCurrentArXiv}.
Let us therefore examine when the condition $B_{\operatorname{eff}}\parallel \Sigma$ is satisfied. Recall (\ref{S4E1}), that $f$ solves the BVP (\ref{S2E3}) and that $\widetilde{\Gamma}_p$ is tangent to $\Sigma$, so that
\begin{gather}
	\nonumber
	\mathcal{N}\cdot B_{\operatorname{eff}}=\mathcal{N}\cdot B+\mathcal{N}\cdot \nabla f
	\\
	\nonumber
	=\left(\frac{\operatorname{Id}}{2}+w^{\operatorname{Tr}}_{\Omega}\right)^{-1}(\mathcal{N}\cdot B)
\end{gather}
and thus
\begin{gather}
	\label{S4E3}
	\mathcal{N}\cdot B_{\operatorname{eff}}=0\Leftrightarrow \mathcal{N}\cdot B=\left(\frac{\operatorname{Id}}{2}+w^{\operatorname{Tr}}_{\Omega}\right)(0)=0
\end{gather}
where we used that $\frac{\operatorname{Id}}{2}+w^{\operatorname{Tr}}_{\Omega}$ is a linear operator in the last step. Therefore, if we want to be able to relate properties of $B$ to the dynamics of $j_p$ we must, according to (\ref{S4E3}), choose the CWS to be an invariant surface of the unique compatible vacuum field.

So suppose now that $B\parallel \Sigma$ (and thus $B_{\operatorname{eff}}\parallel \Sigma$). It then follows that if $B_{\operatorname{eff}}$ does not vanish on $\Sigma$, then neither does $j_p$. In addition, $j_p$ is a div-free field on $\Sigma$ so that it then follows from \cite[Lemma B.1]{Ger25NonoptimalHelicityArXiv} that all field lines of $j_p$ must close poloidally and thus correspond to modular coils.

In practice, if the CWS $\Sigma$ is not chosen as an invariant surface of the compatible vacuum field (which corresponds to the situation $\mathcal{N}\cdot B\neq 0$ on $\Sigma$), then the non-tangency of $B_{\operatorname{eff}}$ to $\Sigma$ may destroy desirable properties which $j_p$ would have inherited from $B_{\operatorname{eff}}$ otherwise. It is therefore of interest in the coil design process to look for a CWS on which $B\cdot \mathcal{N}\approx 0$.

We provide in \Cref{AppC} a sufficient criterion which ensures that $B_{\operatorname{eff}}$ does not vanish on $\Sigma$ if $B\parallel \Sigma$, which in turn, as discussed previously, ensures that all field lines of $j_p$ close poloidally.
\section{Discussion}
\label{Section5}
\subsection{Comparison to previous formulae}
To the best of my knowledge only one other analytic formula has been found, cf. \cite[Lemma 2.8]{Gerner26KernelImage}, which applies for arbitrary target fields, defined on the whole finite region $\Omega$ bounded by the CWS $\Sigma$. More precisely, if $B_T$ is a vector field on $\Omega$ which satisfies $$\operatorname{div}(B_T)=0=\operatorname{curl}(B_T)$$ and $h$ is a solution to the following BVP on $\Omega$
\begin{gather}
	\label{S5E1}
	\Delta h=0\text{, }\mathcal{N}\cdot \nabla h=\left(\frac{\operatorname{Id}}{2}+w^{\operatorname{Tr}}_{\Omega}\right)^{-1}(B_T\cdot \mathcal{N}),
\end{gather}
then
\begin{gather}
	\nonumber
	\operatorname{BS}_{\Sigma}(j)(x)=B_T(x)\text{ in }\Omega\text{ where }
	\\
	\label{S5E2}
	j=\left(\int_{\sigma_t}B_T\right)\Gamma_t\times \mathcal{N}+\nabla h\times \mathcal{N}
\end{gather}
where $\sigma_t$ is any fixed toroidal curve in $\Sigma$ and $\Gamma_t$ is the unique element in $\mathcal{H}_N(\Omega)$ satisfying $\int_{\sigma_t}\Gamma_t=1$.

We emphasise that the formulation in \Cref{S2T7} has several advantageous over the formulation in (\ref{S5E2}). From a computational point of view, in order to be able to compute the current $j$ according to the formula in (\ref{S5E2}), we need to compute an element of $\mathcal{H}_N(\Omega)$ and solve the BVP (\ref{S5E1}). The BVP (\ref{S5E1}) also includes the inversion of the operator $\frac{\operatorname{Id}}{2}+w^{\operatorname{Tr}}_{\Omega}$ so that solving a BVP of the form (\ref{S5E1}) is comparable in complexity to solving the BVP (\ref{S2E3}). A minor disadvantage of the formulation in \Cref{S2T7} is that we need to compute an element of $\mathcal{H}_N(\widetilde{\Omega})$ which is defined on an unbounded domain. However, we have seen in \Cref{S3P1}, that one can reduce this problem to finite domains, so that the computational complexity of computing an element in $\mathcal{H}_N(\Omega)$ and $\mathcal{H}_N(\widetilde{\Omega})$ is comparable. The main difference is that the formulation in (\ref{S5E1}) and (\ref{S5E2}) works not directly with the plasma equilibrium field $B$ (and its corresponding vacuum extension) but it works with the corresponding target field $B_T=B-\operatorname{BS}_P(J)$ and we require knowledge of this field on $\Sigma$. This means that in formulation (\ref{S5E2}) we also need to compute the compatible vacuum field, but in addition we also need to compute the quantity $\operatorname{BS}_P(J)$. While one can avoid singularities of the integral kernel by noticing that (\ref{S5E1}) only requires knowledge 
of $\operatorname{BS}_P(J)$ on\slash near $\Sigma$ which is of positive distance to $P$ and noticing that $\sigma_t$ is also contained in $\Sigma$, the formulation in (\ref{S5E2}) adds an additional computation which may also introduce numerical inaccuracies which are avoided in the formulation in \Cref{S2T7}. From a more theoretical point of view, the main advantage of \Cref{S2T7} in comparison to (\ref{S5E2}) is that it works directly with the full plasma equilibrium field rather than with the corresponding target field. This allows one to infer more directly which properties of the full plasma equilibrium field $B$ are relevant for coil design complexity. In particular, it is quite natural to try to find a CWS $\Sigma$ for which $B\parallel \Sigma$, because in leading order charged particles follow magnetic field lines, cf. \cite{littlejohn1983variational}, so that if $B$ is not tangent to $\Sigma$ and some particle escapes the plasma region $P$, then the field lines of $B$ would lead them directly to $\Sigma$ which may cause structural damage. It follows further from the BVP (\ref{S2E3}) that if $B\parallel \Sigma$, then $\nabla f=0$ and so the current formula simplifies significantly. One should further expect all field lines of the resulting current $j_p$ to close poloidally. In contrast, in order for the formula (\ref{S5E2}) to simplify in the same way one would require, according to (\ref{S5E1}), $B_T\parallel \Sigma$. However, this condition does not imply any advantageous features of the full plasma equilibrium field and in this sense is unnatural.
\subsection{Coil geometry and plasma equilibrium properties}
In recent work, \cite{RodSengArXiv}, a simplified situation has been considered where the full plasma equilibrium field $B$ is assumed to be a vacuum field, i.e. to satisfy $\operatorname{curl}(B)=0$ in $P$, and it was additionally assumed that the CWS coincides with the plasma boundary $\partial P$ in order to gain some insights regarding which properties of $B$ influence the geometric complexity of the field lines of the corresponding current $j$ reproducing $B$. In this case the current $j:=B\times \mathcal{N}$ reproduces $B$ in $P=\Omega$ (recall we assume here for the moment that $\Sigma=\partial P$), cf. \cite[Equation (2)]{RodSengArXiv}. We notice that if $B$ is curl-free, then $\int_{\sigma_p}B=0$ by Stokes theorem since $\sigma_p\subset \Sigma=\partial P$ bounds a disc in $P$. In addition, $B\cdot \mathcal{N}=0$ on $\partial P$ by the usual boundary condition assumption of $B$ so that the formula $j=B\times \mathcal{N}$ is consistent with \Cref{S2T7} and the choice of $\alpha$ which makes the field lines of the current poloidal.

To extend the discussion in \cite{RodSengArXiv} to non-vacuum equilibrium fields and CWS of positive distance to the plasma region $P$ suppose from now on that $B$ is not necessarily curl-free, i.e. it is allowed to be a finite $\beta$-plasma, and that $\operatorname{dist}(P,\Sigma)>0$. We have already discussed that it is advantageous to seek a CWS $\Sigma$ which satisfies $B\parallel \Sigma$ (recall that outside of $P$, $B$ is defined to be its unique compatible vacuum extension) because under mild assumptions this leads to poloidal field lines of the coil current $j_p$ which then takes the form
\begin{gather}
	\label{S5E3}
	j_p=B_{\operatorname{eff}}\times \mathcal{N}=B\times \mathcal{N}-\left(\int_{\sigma_p}B\right)\widetilde{\Gamma}_p\times \mathcal{N}.
\end{gather}
We hence get $B_{\operatorname{eff}}=B-\left(\int_{\sigma_p}B\right)\widetilde{\Gamma}_p$. The formula for the normal curvature $\widetilde{\kappa}_n$ and the geodesic curvature $\widetilde{\kappa}_g$ of $j_p$ are then given by the formulas, cf. \cite[Equation (5)]{RodSengArXiv} (we adapt the notation from \cite{RodSengArXiv})
\begin{gather}
	\label{S5E4}
	\widetilde{\kappa}_n=2H-\kappa_n\text{ and }\widetilde{\kappa}_g=\frac{\hat{\tau}\cdot \nabla_{\hat{\tau}}B_{\operatorname{eff}}}{|B_{\operatorname{eff}}|}
\end{gather}
where $H=\frac{\kappa_1+\kappa_2}{2}$ is the mean curvature of $\Sigma$ (and $\kappa_1,\kappa_2$ are the principal curvatures w.r.t. $\mathcal{N}$), $\kappa_n$ denotes the normal curvature of $B_{\operatorname{eff}}$ and $\hat{\tau}:=\frac{B_{\operatorname{eff}}\times \mathcal{N}}{|B_{\operatorname{eff}}|}$. Now the discussions in \cite[Section III]{RodSengArXiv} apply with two caveats
\begin{enumerate}
	\item It is in general not the magnetic field $B$ itself that matters, but instead the current complexity is determined by the effective magnetic field $B_{\operatorname{eff}}=B-\left(\int_{\sigma_p}B\right)\widetilde{\Gamma}_p$. Notice however that the vacuum extension is uniquely determined by $B|_{\partial P}$ and thus so are the possible choices of $\Sigma$ (if we insist that the CWS is an invariant surface of the vacuum field). In addition, $\widetilde{\Gamma}_p$ is uniquely determined by $\Sigma$ so that the effective field $B_{\operatorname{eff}}$ is uniquely determined by the two input parameters $(B|_{\partial P},\operatorname{dist}(\Sigma,P))$ (again assuming the CWS is chosen as an invariant surface).
	\item If the CWS is far away from $P$, then the behaviour of the corresponding vacuum extension, i.e. $B|_{\Sigma}$, might differ greatly from the behaviour of $B|_{\partial P}$. So even if $\int_{\sigma_p}B=0$ and thus $B_{\operatorname{eff}}=B$ it is of importance to understand in more detail how the dynamics of $B|_{\partial P}$ influence the dynamics of its vacuum field extension in regions which are far away from the plasma region. We leave this question open as a potential future research direction.
\end{enumerate}
\subsection{Current strength and magnetic normal curvature}
It is customary in the stellarator community to regard $\|j\|_{L^{\infty}(\Sigma)}$ as a proxy for coil complexity because it is related to the minimal coil-coil distance and number of coils, cf. \cite[Section 3.1]{KLM24}. Following this idea it has been found empirically in \cite[Section 4]{KLM24} that if $y=x+\delta \mathcal{N}(x)$, for $x\in \partial P$ and $\delta=\operatorname{dist}(y,\partial P)$, is a point on a conformal CWS $\Sigma_{\delta}$ at which the maximum of $|j(y)|$ is achieved, then the corresponding point $x\in \partial P$ on the plasma boundary tends to be a point of smallest magnetic curvature radius, or equivalently a point of highest magnetic normal curvature.

We provide here a theoretical explanation for this phenomenon. We start by observing that in a sufficiently small neighbourhood around the plasma boundary $\partial P$, every point $y$ admits a unique point $x\in \partial P$ such that $\operatorname{dist}(y,\partial P)=\operatorname{dist}(y,x)$, and in this case we have the identity $y=x+\delta \mathcal{N}(x)$ where $\delta=\operatorname{dist}(y,\Sigma)$. As alluded to in Section 3.1 we can compute the unique vacuum field compatible with a given plasma equilibrium field $B$ in such a neighbourhood by means of a power series expansion. Following the procedure outlined in \cite[Appendix D]{Ger26CoilGeometry} we obtain the following expansion for the vacuum field $B$
\begin{gather}
	\nonumber
	B(y)=B(x)+\delta s(B)(x)
	\\
	\label{S5E5}
	- \delta \operatorname{div}_{\partial P}(B)(x)\mathcal{N}(x)+\mathcal{O}(\delta^2)\text{ as }\delta\searrow 0
\end{gather}
where $y=x+\delta \mathcal{N}(x)$ for the unique projection point $x\in \partial P$, $\delta=\operatorname{dist}(y,\partial P)$, $s$ denotes the shape operator on $\partial P$ (recall that we assume $B\parallel \partial P$ and thus it makes sense to apply the shape operator to $B$ on $\partial P$) and where $\operatorname{div}_{\partial P}(B)=\operatorname{div}(B)-\mathcal{N}\cdot \nabla_{\mathcal{N}}B=-\mathcal{N}\cdot \nabla_{\mathcal{N}}B$ is the intrinsic divergence of $B$ on $\partial P$.

Now, keeping in mind that $B(x)$ and $s(B)(x)$ are both tangent to $\partial P$, we obtain the following Taylor expansion for $|B(y)|$
\begin{gather}
	\nonumber
	|B(y)|=|B(x)|\left(1+b(x)\cdot s(b)(x)\delta\right)+\mathcal{O}(\delta^2)
	\\
	\label{S5E6}
	=|B(x)|\left(1+\kappa_n(x)\delta\right)+\mathcal{O}(\delta^2)
\end{gather}
as $\delta\searrow 0$, where $b(x)=\frac{B(x)}{|B(x)|}$ and $\kappa_n(x)$ denotes the normal curvature of the magnetic field $B$ at position $x$.

To gain more insight we now consider the situation in which the CWS $\Sigma$ is chosen such that $B\parallel \Sigma$ and such that $\int_{\sigma_p}B=0$. In this case the effective magnetic field reduces to the magnetic field $B$ itself, $B_{\operatorname{eff}}=B$, and thus our poloidal current $j_p$, recall (\ref{S4E2}), becomes
$$j_p(y)=B(y)\times \mathcal{N}(y)\text{ on }\Sigma.$$
In this case we find according to (\ref{S5E6})
\begin{gather}
	\label{S5E7}
	|j_p(y)|=|B(x)|\left(1+\kappa_n(x)\delta\right)+\mathcal{O}(\delta^2).
\end{gather}
It is now immediately clear from (\ref{S5E7}) that regions of high (positive) $\kappa_n$ favour high current strength which is in agreement with the empirical observations made in \cite{KLM24} regarding conformal CWSs. One additional feature stemming from the fact that the CWS under consideration is chosen as an invariant surface of the compatible vacuum field rather than as a conformal surface is that distinct points $y$ on the CWS may correspond to different values of $\delta$. Thus, (\ref{S5E7}) tells us that one should expect to find stronger currents in regions which are farther away from the plasma boundary, i.e. points $y$ which correspond to larger values of $\delta$.  

Lastly, we point out that (\ref{S5E7}) suggests that the maximum of $|j_p(y)|$ will be achieved at a point at which $\kappa_n(x)$ is positive. We then see that at each such $x$ we formally have
\begin{gather}
	\nonumber
	\frac{d|j_p(x+\delta \mathcal{N}(x))|}{d\delta}|_{\delta=0}
	\\
	\nonumber
	=|B(x)|\kappa_n(x)>0
\end{gather}
and thus the current strength increases monotonically with the coil-plasma distance $\delta$. This explains the behaviour observed in \cite[Figure 4]{KLM24} where precisely this monotonicity was noticed and used to associate a unique coil-plasma distance to a prescribed value of $\|j_p\|_{L^{\infty}(\Sigma)}$.

We recall here that we assumed throughout this discussion that $\int_{\sigma_p}B=0$. In the more general setting $\int_{\sigma_p}B\neq 0$, the current strength is given by $|j_p(y)|=|B_{\operatorname{eff}}(y)|$ and the quantity $\left(\int_{\sigma_p}B\right)\widetilde{\Gamma}_p$ enters the equation. To generalise the above discussion it would therefore be of essence to understand in more detail the behaviour of the vector field $\widetilde{\Gamma}_p$ on $\Sigma$. This is left as a future research direction.
\subsection{Conformal surfaces}
In practice one often considers CWSs which are conformal to the plasma boundary, see e.g. \cite{KLM24}. I.e. one considers surfaces which, for a fixed distance $\delta>0$, are defined by $$\Sigma_{\delta}:=\{x+\delta \mathcal{N}(x)\mid x\in \partial P\}.$$
These type of surfaces, while common, do not need to be invariant surfaces of the vacuum extension of a given plasma equilibrium field $B$. In the near plasma boundary limit one has $$|\mathcal{N}_{\Sigma_{\delta}}\cdot B|\in \mathcal{O}\left(\delta\right)$$ because $B\cdot \mathcal{N}=0$ on $\partial P$. It follows in fact from (\ref{S5E5}) that
\begin{gather}
	\nonumber
	|\mathcal{N}_{\Sigma_{\delta}}\cdot B|\in \mathcal{O}\left(\delta^2\right)
	\Leftrightarrow \operatorname{div}_{\partial P}(B)=0\text{ on }\partial P
\end{gather}
because $\mathcal{N}_{\Sigma_{\delta}}(x+\delta \mathcal{N}(x))=\mathcal{N}(x)$. It follows further from the fact that $B$ satisfies (\ref{S1EExtra}), that
\begin{gather}
	\nonumber
	\operatorname{curl}_{\partial P}(B)=\mathcal{N}\cdot \operatorname{curl}(B)=0\text{ on }\partial P.
\end{gather}
So a conformal surface is an invariant surface of the compatible vacuum field up to second order if and only if $B$ is an intrinsically harmonic vector field on the plasma boundary.

In any case, if a conformal CWS is not too far from the plasma boundary, then the compatible vacuum field is approximately tangent to the CWS and the correction term $\nabla f\times \mathcal{N}$ in the formula for $j$, cf. \Cref{S2T7}, (\ref{S2E3}), will not influence the dynamics too much. In particular if $\operatorname{dist}(\Sigma_{\delta},P)$ is small enough, then one should still expect that all field lines of the resulting currents remain poloidal. On the other hand, if the distance becomes large, then the correction term $\nabla f\times \mathcal{N}$ may start to be of the same order of magnitude as $B\times \mathcal{N}-\left(\int_{\sigma_p}B\right)\widetilde{\Gamma}_p\times \mathcal{N}$ which may introduce zeros to the current distribution which in turn lead to centre and saddle singularity regions on the CWS and increases current\slash coil complexity. This phenomenon is well-known in standard coil optimisation procedures, cf. \cite[Figure 3]{L17}, \cite[Figure 1]{KLM24}. \Cref{S2T7} suggests that if the CWS is not selected as a conformal surfaces, but instead as an invariant surface of the compatible vacuum field of the plasma equilibrium, then this might simplify the current patterns and thus the coil design.
\section{Summary}
\label{Section6}
In this work we provided an analytic formula which, for a given plasma equilibrium field and CWS, computes a current distribution on the CWS which supports the given plasma equilibrium field and its unique compatible vacuum extension. The provided formula works with the full equilibrium field directly and thus opens the door to a more thorough theoretical understanding regarding which properties of a plasma equilibrium field influence the coil complexity. To illustrate this we used this new formula to provide a theoretical explanation of the empirical observation that currents on a CWS tend to achieve their maximal strength at points which correspond to the highest normal magnetic curvature on the plasma boundary. We compared our findings with previous literature and also discussed potential avenues regarding the approximation of the quantities appearing in our analytic formula. Some open questions remain. In particular, the question how precisely the dynamics of the magnetic field lines on the plasma boundary influence the dynamics of its unique vacuum field extension at larger distances. Answering this question will give further insight into which properties of a plasma equilibrium field determine the coil complexity.
\end{multicols}
\section*{Data availability}
No new data was created during the course of this work.
\section*{Conflict of interest}
The author declares that he has no conflict of interest.
\section*{Acknowledgements}
The research was supported in part by the MIUR Excellence Department Project awarded to Dipartimento di Matematica, Università di Genova, CUP D33C23001110001.
\appendix
\section{Exterior Neumann fields and the kernel of the Biot Savart operator}
\label{AppA}
In this section we prove \Cref{S2P4}, \Cref{S2T5} and \Cref{S3P1}. Throughout this section $\Sigma$ is a smooth toroidal surface in $\mathbb{R}^3$. We point out that while the relationship between the dimension of the space of Harmonic Neumann fields and the topology of the underlying domain is well-established for bounded domains as well as weighted exterior domains \cite[Section 2.4, Section 2.5 \& Section 2.6]{S95}, the definition provided in \Cref{S2D3} does not quite fit these frameworks. We therefore provide a proof.
\begin{proof}[Proof of \Cref{S2P4}]
	We recall that we denote by $\widetilde{\Omega}$ the unbounded component of $\mathbb{R}^3\setminus \Sigma$ and that $\mathcal{H}_N(\widetilde{\Omega})$ denotes the vector space of all (smooth) vector fields on $\widetilde{\Omega}$ which are curl- and div-free, tangent to $\partial \widetilde{\Omega}=\Sigma$ and vanish at infinity. The first part of \Cref{S2P4} claims that $\operatorname{dim}\left(\mathcal{H}_N(\widetilde{\Omega})\right)=1.$
	\newline
	\newline
	\underline{Step 1: $\operatorname{dim}\left(\mathcal{H}_N(\widetilde{\Omega})\right)\geq 1$:} Denote by $\Omega$ the finite volume enclosed by $\Sigma$. It is well-known, cf. \cite[Hodge decomposition theorem]{CDG02}, that $\mathcal{H}_N(\Omega)$ is $1$-dimensional. We may then fix some $\Gamma\in \mathcal{H}_N(\Omega)\setminus \{0\}$ and define $$\operatorname{BS}_{\Omega}(\Gamma)(x):=\frac{\int_{\Omega}\Gamma(y)\times \frac{x-y}{|x-y|^3}d^3y}{4\pi}\text{, }x\in \mathbb{R}^3.$$
	Since $\operatorname{BS}_{\Omega}(\Gamma)$ is div-free in $\mathbb{R}^3$ it follows from Gauss' theorem that $\int_{\partial\Omega}\mathcal{N}\cdot \operatorname{BS}_{\Omega}(\Gamma) dS=0$. Hence, there exists a unique solution $g\in \left\{f\in L^6(\widetilde{\Omega})\mid \nabla f\in L^2(\widetilde{\Omega})\right\}$ of the following boundary value problem
	\begin{AppA}
		\label{AppAE1}
		\Delta g=0\text{ in }\widetilde{\Omega}\text{, }\mathcal{N}\cdot \nabla g=-\mathcal{N}\cdot \operatorname{BS}_{\Omega}(\Gamma)\text{, }g(x)\rightarrow 0\text{ as }|x|\rightarrow\infty.
	\end{AppA}
	Now consider
	\begin{AppA}
		\label{AppAE2}
		\widetilde{\Gamma}:=\operatorname{BS}_{\Omega}(\Gamma)+\nabla g.
	\end{AppA}
	Observe that $\operatorname{div}(\widetilde{\Gamma})=0=\operatorname{curl}(\widetilde{\Gamma})$ and $\mathcal{N}\cdot \widetilde{\Gamma}=0$. Further, if $\sigma_p$ is any (poloidal) loop on $\partial \widetilde{\Omega}=\Sigma=\partial \Omega$ which bounds a surface $\mathcal{A}$ in $\Omega$, then it is a standard fact that $\int_{\sigma_p}\operatorname{BS}_{\Omega}(\Gamma)=\int_{\mathcal{A}}\Gamma\cdot n=\operatorname{Flux}(\Gamma)\neq 0$ and thus $\operatorname{BS}_{\Omega}(\Gamma)$ is not a gradient field on $\widetilde{\Omega}$. Consequently $\widetilde{\Gamma}$ is not identically zero. Finally, we notice that $\operatorname{BS}_{\Omega}(\Gamma)\in \mathcal{O}(|x|^{-3})$, \cite[Section III. A.]{CDG01} and that the unique solution of the BVP (\ref{AppAE1}) admits the following representation, cf. \cite[Theorem 6.43]{RCM21},
	\begin{AppA}
		\label{AppAE3}
		g(x)=\int_{\Sigma}\frac{\left(\frac{\operatorname{Id}}{2}+w^{\operatorname{Tr}}_{\Omega}\right)^{-1}(\mathcal{N}\cdot \operatorname{BS}_{\Omega}(\Gamma))(y)}{4\pi |x-y|}dS(y).
	\end{AppA}
	Consequently $$\nabla g(x)=-\int_{\Sigma}\frac{\left(\frac{\operatorname{Id}}{2}+w^{\operatorname{Tr}}_{\Omega}\right)^{-1}\left(\mathcal{N}\cdot \operatorname{BS}_{\Omega}(\Gamma)\right)(y)}{4\pi}\frac{x-y}{|x-y|^3}dS(y).$$ It then follows mutatis mutandis from the arguments in the proof of \cite[Proof of Appendix Lemma 4]{CDG01} that $|\nabla g(x)|\in \mathcal{O}(|x|^{-3})$ as $|x|\rightarrow\infty$. Overall, $0\neq \widetilde{\Gamma}\in \mathcal{O}(|x|^{-3})$ is a non-zero element of $\mathcal{H}_N(\widetilde{\Omega})$ and thus $\operatorname{div}\left(\mathcal{H}_N(\widetilde{\Omega})\right)\geq 1$.
	\newline
	\newline
	\underline{Step 2: $\operatorname{dim}\left(\mathcal{H}_N(\widetilde{\Omega})\right)\leq 1$:} Now suppose that $\widetilde{\Gamma}_0\in \mathcal{H}_N(\widetilde{\Omega})$. One can then follow the reasoning in the derivation of the virtual-casing principle as in \cite[Proof of Lemma 5.5]{G24} in order to show that
	\begin{AppA}
		\nonumber
		\operatorname{BS}_{\Sigma}(\widetilde{\Gamma}_0\times \mathcal{N})(x)=\frac{\int_{\Sigma}(\widetilde{\Gamma}_0(y)\cdot \mathcal{N}(y))\frac{x-y}{|x-y|^3}dS(y)}{4\pi}-\frac{\int_{\widetilde{\Omega}}\operatorname{curl}(\widetilde{\Gamma}_0)(y)\times \frac{x-y}{|x-y|^3}d^3y}{4\pi}
		\end{AppA}
		\begin{AppA}
			\label{AppAExtraExtraExtra}
		-\frac{\int_{\widetilde{\Omega}}\operatorname{div}(\widetilde{\Gamma}_0)(y)\frac{x-y}{|x-y|^3}d^3y}{4\pi}+\begin{cases}
			\widetilde{\Gamma}_0(x) &\text{ if }x\in \widetilde{\Omega} \\ 0 & \text{ if }x\in \Omega
		\end{cases}
	\end{AppA}
	where terms at infinity may be discarded during the derivation due to the behaviour of $\widetilde{\Gamma}_0$ at infinity and where $\mathcal{N}$ denotes the outward pointing unit normal on $\Sigma$ with respect to the exterior domain $\widetilde{\Omega}$. We conclude from (\ref{AppAExtraExtraExtra}) that $\operatorname{BS}_{\Sigma}(\widetilde{\Gamma}_0\times \mathcal{N})=0$ in $\Omega$. However, it also follows that $\operatorname{BS}_{\Sigma}(\widetilde{\Gamma}\times \mathcal{N})=0$ in $\Omega$ where $\widetilde{\Gamma}$ is a non-zero element of $\mathcal{H}_N(\widetilde{\Omega})$ as constructed in (\ref{AppAE2}). It then additionally follows from \cite[Theorem 5.1]{G24} that the kernel of the Biot-Savart operator (within the interior domain $\Omega$) is $1$-dimensional. Thus, there exists some $\alpha\in \mathbb{R}$ such that $$\widetilde{\Gamma}_0\times \mathcal{N}=\alpha \widetilde{\Gamma}\times \mathcal{N}\Leftrightarrow \left(\widetilde{\Gamma}_0-\alpha \widetilde{\Gamma}\right)\times \mathcal{N}=0\text{ on }\Sigma.$$
	Since $\widetilde{\Gamma}_0-\alpha \widetilde{\Gamma}\in \mathcal{H}_N(\widetilde{\Omega})$ we conclude that its normal trace also vanishes on $\Sigma$ and consequently $$\widetilde{\Gamma}_0-\alpha \widetilde{\Gamma}=0\text{ on }\Sigma.$$
	It then follows from \cite[Theorem 3.4.4]{S95} and the fact that $\widetilde{\Gamma}_0-\alpha \widetilde{\Gamma}$ is div- and curl-free that $$\widetilde{\Gamma}_0=\alpha \widetilde{\Gamma}\text{ in }\widetilde{\Omega}.$$
	This proves that $\operatorname{dim}\left(\mathcal{H}_N(\widetilde{\Omega})\right)\leq 1$.
	\newline
	\newline
	\underline{Step 3: Concluding the proof:} We are left with arguing that for any fixed poloidal loop $\sigma_p$ on $\Sigma$ there exists a unique element $\widetilde{\Gamma}_p\in \mathcal{H}_N(\widetilde{\Omega})$ with $\int_{\sigma_p}\widetilde{\Gamma}_p=1$. It follows from the fact that the space $\mathcal{H}_N(\widetilde{\Omega})$ is $1$-dimensional that if for some $\widetilde{\Gamma}\in \mathcal{H}_N(\widetilde{\Omega})$ we have $\int_{\sigma_p}\widetilde{\Gamma}\neq 0$, then by scaling there is a unique $\widetilde{\Gamma}_p\in \mathcal{H}_N(\widetilde{\Omega})$ with $\int_{\sigma_p}\widetilde{\Gamma}_p=1$. But we have already seen in the first step that according to (\ref{AppAE2}) there exists some non-zero element $\Gamma\in \mathcal{H}_N(\Omega)$ such that for $\widetilde{\Gamma}\in \mathcal{H}_N(\widetilde{\Omega})$ given by (\ref{AppAE2}) we have $\int_{\sigma_p}\widetilde{\Gamma}=\operatorname{Flux}(\Gamma)\neq 0$ which completes the proof.
\end{proof}
We point out that we have additionally shown
\begin{AppA}
	\label{AppAE4}
	\mathcal{H}_N(\widetilde{\Omega})=\left\{\Gamma\in H^1(\widetilde{\Omega})\mid \operatorname{curl}(\Gamma)=0=\operatorname{div}(\Gamma)\text{, }\Gamma\parallel \partial\widetilde{\Omega}\text{, }|\Gamma(x)|=\mathcal{O}(|x|^{-3})\text{ as }|x|\rightarrow\infty\right\}.
\end{AppA}
This enables us to prove the specific convergence rate claimed in \Cref{S3P1}. We point out for completeness that the convergence $\Gamma_r\times \mathcal{N}\rightarrow \widetilde{\Gamma}_p\times \mathcal{N}$ takes place in $W^{-\frac{1}{2},2}(\Sigma)$ topology.
\begin{proof}[Proof of \Cref{S3P1}]
	Fix $0<r\leq R$ such that $\overline{\Omega}\subset B_r(0)$ and set $U_R:=B_R(0)\setminus \overline{\Omega}$. We denote by $\Gamma_R$ the unique element of $\mathcal{H}_N(U_R)$ satisfying $\int_{\sigma_p}\Gamma_R=1$ for some fixed poloidal curve $\sigma_p$ on $\partial \Omega=\Sigma=\partial \widetilde{\Omega}$ and denote by $\widetilde{\Gamma}_p$ the unique element of $\mathcal{H}_N(\widetilde{\Omega})$ satisfying $\int_{\sigma_p}\widetilde{\Gamma}_p=1$. We observe that the restriction $\widetilde{\Gamma}_p|_{U_R}$ is curl- and div-free and has the same poloidal circulation as $\Gamma_R$ along $\sigma_p$. Then the Hodge decomposition theorem implies that $$\widetilde{\Gamma}_p=\Gamma_R+\nabla f_R\text{ on }U_R$$
	for a suitable function $f_R\in H^2(U_R)$ (to uniquely fix $f_R$ we assume $\int_{U_R}f_Rd^3x=0$). Since this decomposition is $L^2(U_R)$-orthogonal we find
	\begin{AppA}
		\label{AppAE5}
		\|\nabla f_R\|^2_{L^2(U_R)}=\int_{U_R}\widetilde{\Gamma}_p\cdot \nabla f_Rd^3x=\int_{\partial B_R(0)}(\mathcal{N}\cdot \widetilde{\Gamma}_p)f_Rdx
	\end{AppA}
	where we integrated by parts and used that $\operatorname{div}(\widetilde{\Gamma}_p)=0$ and $\widetilde{\Gamma}_p\parallel \Sigma$ in the last step. We now define the following homogeneous version of the fractional Sobolev norm
	\begin{AppA}
		\label{AppAE6}
		\|f\|^2_{W^{\frac{1}{2},2}(\partial B_R(0))}:=\|f\|^2_{L^2(\partial B_R(0))}+R\int_{\partial B_R(0)}\int_{\partial B_R(0)}\frac{|f_R(x)-f_R(y)|^2}{|x-y|^3}dS(x)dS(y).
	\end{AppA}
	Therefore, if we denote by $\|\cdot\|_{W^{-\frac{1}{2},2}(\partial B_R(0))}$ the corresponding dual norm, then (\ref{AppAE5}) becomes
	\begin{AppA}
		\label{AppAE7}
		\|\nabla f_R\|^2_{L^2(U_R)}\leq \|\mathcal{N}\cdot \widetilde{\Gamma}_p\|_{W^{-\frac{1}{2},2}(\partial B_R(0))}\|f_R\|_{W^{\frac{1}{2},2}(\partial B_R(0))}.
	\end{AppA}
	It then follows from a scaling argument and the trace inequality applied to $\partial B_1(0)$ that there is a constant $c>0$ (independent of $R$) such that if $F_R$ is any $H^1(B_R(0))$ extension of $f_R$, then
	\begin{gather}
		\nonumber
		\|f_R\|^2_{W^{\frac{1}{2},2}(\partial B_R(0))}\leq c\left(\frac{\|F_R\|^2_{L^2(B_R(0))}}{R}+R\|\nabla F_R\|^2_{L^2(B_R(0))}\right).
	\end{gather}
	Now let $\rho_R$ be a smooth bump function with $0\leq \rho_R\leq 1$, $\rho_R=1$ on $\partial B_R(0)$ and $\operatorname{supp}(\rho_R)\cap \Omega=\emptyset$. We can then consider the extension $F_R:=\rho_Rf_R$ and using that $\rho_R$ has no support in $\Omega$ we find
	\begin{gather}
		\nonumber
		\|f_R\|^2_{W^{\frac{1}{2},2}(\partial B_R(0))}\leq c\left(\frac{\|f_R\|^2_{L^2(U_R)}}{R}+R\|\nabla f_R\|^2_{L^2(U_R)}+R\|f_R\nabla \rho_R\|^2_{L^2(U_R)}\right)
	\end{gather}
	where from now on $c>0$ denotes a generic constant which may change its value from line to line but is always independent of $R$. We now notice that it is possible to select $\rho_R$ such that $\|\nabla \rho_R\|_{L^{\infty}(U_R)}\leq \frac{c}{\operatorname{dist}(\partial B_R,\partial \widetilde{\Omega})}$ and so for large enough $R$ we get $\|\nabla \rho_R\|_{L^{\infty}(U_R)}\leq \frac{c}{R}$. We therefore arrive at the estimate
	\begin{gather}
		\nonumber
		\|f_R\|^2_{W^{\frac{1}{2},2}(\partial B_R(0))}\leq c\left(\frac{\|f_R\|^2_{L^2(U_R)}}{R}+R\|\nabla f_R\|^2_{L^2(U_R)}\right).
	\end{gather}
	Now, since we assume $\int_{U_R}f_Rd^3x=0$, we obtain from the Poincar\'{e}-Wirtinger inequality
	\begin{AppA}
		\label{AppAE11}
		\|f_R\|^2_{W^{\frac{1}{2},2}(\partial B_R(0))}\leq cR\|\nabla f_R\|^2_{L^2(U_R)}\left(1+\frac{1}{\lambda_N(U_R)R^2}\right)
	\end{AppA}
	where $\lambda_N(U_R)$ is the first non-zero Neumann eigenvalue of $U_R$. We then observe that by the scaling properties of the Neumann eigenvalues we have $$\frac{\lambda_N(U_R)R^2}{\lambda_N(B_R(0))R^2}=\frac{\lambda_N\left(B_1(0)\setminus \frac{\Omega}{R}\right)}{\lambda_N(B_1(0))}\rightarrow 1\text{ as }R\rightarrow\infty$$
	where $\frac{\Omega}{R}=\left\{\frac{x}{R}\mid x\in \Omega\right\}$ and the convergence follows from \cite[Theorem 3.1]{RauTay75}. Thus, $\lambda_N(U_R)R^2\geq c>0$ because $\lambda_N(B_R(0))R^2=\lambda_N((B_1(0)))$ is positive and consequently (\ref{AppAE11}) becomes
	\begin{gather}
		\nonumber
		\|f_R\|^2_{W^{\frac{1}{2},2}(\partial B_R(0))}\leq cR\|\nabla f_R\|^2_{L^2(U_R)}.
	\end{gather}
	We can insert this into (\ref{AppAE7}) and find
	\begin{AppA}
		\label{AppAE12}
		\|\nabla f_R\|_{L^2(U_R)}\leq c\sqrt{R}\|\mathcal{N}\cdot \widetilde{\Gamma}_p\|_{W^{-\frac{1}{2},2}(\partial B_R(0))}.
	\end{AppA}
	Finally, we observe that by definition of the dual norm we have
	\begin{gather}
		\nonumber
		\|\mathcal{N}\cdot \widetilde{\Gamma}_p\|_{W^{-\frac{1}{2},2}(\partial B_R(0))}=\sup_{h\in W^{\frac{1}{2},2}(\partial B_R(0))\setminus \{0\}}\frac{\int_{\partial B_R}(\mathcal{N}\cdot \widetilde{\Gamma}_p)hdS}{\|h\|_{W^{\frac{1}{2},2}(\partial B_R)}}.
	\end{gather}
	We then estimate for any fixed $h$
	$$\int_{\partial B_R(0)}(\mathcal{N}\cdot \widetilde{\Gamma}_p)hdS\leq \max_{\partial B_R(0)}|\widetilde{\Gamma}_p|\int_{\partial B_R(0)}|h|dS\leq cR^{-3}\|h\|_{L^1(\partial B_R(0))}$$
	where we used that $|\widetilde{\Gamma}_p(x)|\in \mathcal{O}(|x|^{-3})$ as $|x|\rightarrow\infty$. Then a scaling argument yields
	\begin{gather}
		\nonumber
		\sup_{h\in W^{\frac{1}{2},2}(\partial B_R(0))\setminus \{0\}}\frac{\|h\|_{L^1(\partial B_R(0))}}{\|h\|_{W^{\frac{1}{2},2}(\partial B_R(0))}}=R\sup_{h\in W^{\frac{1}{2},2}(\partial B_1(0))\setminus \{0\}}\frac{\|h\|_{L^1(\partial B_1(0))}}{\|h\|_{W^{\frac{1}{2},2}(\partial B_1(0))}}.
	\end{gather}
	We observe that $$\|h\|_{L^1(\partial B_1(0))}\leq \sqrt{4\pi}\|h\|_{L^2(\partial B_1(0))}\leq \sqrt{4\pi}\|h\|_{W^{\frac{1}{2},2}(\partial B_1(0))}$$ and so inserting our estimate for the dual norm into (\ref{AppAE12}) yields
	\begin{gather}
		\nonumber
		\|\widetilde{\Gamma}_p-\Gamma_R\|_{L^2(U_R)}=\|\nabla f_R\|_{L^2(U_R)}\leq cR^{-\frac{3}{2}}\text{ as }R\rightarrow\infty.
	\end{gather}
	Since $\operatorname{curl}(\widetilde{\Gamma}_p-\Gamma_R)=0$, it follows that $\|\widetilde{\Gamma}_p-\Gamma_R\|_{H(U_R,\operatorname{curl})}\in \mathcal{O}\left(R^{-\frac{3}{2}}\right)$ and so the continuity of the twisted tangential trace with respect to this norm, cf. \cite[Theorem 2.11]{GR86}, implies that $$\|\widetilde{\Gamma}_p\times \mathcal{N}-\Gamma_R\times \mathcal{N}\|_{W^{-\frac{1}{2},2}(\Sigma)}\in \mathcal{O}\left(R^{-\frac{3}{2}}\right).$$
	This completes the proof.
\end{proof}
We now come to the characterisation of the kernel of the Biot-Savart operator.
\begin{proof}[Proof of \Cref{S2T5}]
	The fact that $j_0=\mathcal{N}\times \widetilde{\Gamma}_p$ lies in the kernel follows either directly from the virtual-casing principle, cf. (\ref{AppAExtraExtraExtra}), or may be seen as a consequence of the formula (\ref{AppAE2}) in combination with \cite[Theorem 2.13]{Gerner26KernelImage}. The fact that the kernel is $1$-dimensional follows from \cite[Theorem 5.1 \& Remark C.2]{G24}.
\end{proof}
\begin{rem}
	We point out for completeness that the domain of $\operatorname{BS}_{\Sigma}$ in the statement of \Cref{S2T5} is taken to be the completion of the space of square integrable, div-free current distributions on $\Sigma$ with respect to the $W^{-\frac{1}{2},2}$-norm, cf. \cite[Appendix C]{G24}, which in particular includes all square integrable, div-free current distributions on $\Sigma$.
\end{rem}
\section{Proof of the analytic current formula}
\label{AppB}
We assume throughout this section that $\Sigma\subset \mathbb{R}^3$ is a smooth toroidal surface.

Before we come to the proof of the main result \Cref{S2T7} we show that the BVP (\ref{S2E3}) indeed admits a solution, i.e. we first prove \Cref{S2L6}.
\begin{proof}[Proof of \Cref{S2L6}]
	We recall that we denote by $\Omega$ the finite region enclosed by $\Sigma$ and that $B$ is a div-free field which is square integrable. This implies that $\int_{\partial\Omega}\mathcal{N}\cdot BdS=0$ and thus there exists an $H^1(\Omega)$ solution to the auxiliary BVP
	\begin{AppB}
		\label{AppBE1}
		\Delta h=0\text{ in }\Omega\text{, }\mathcal{N}\cdot \nabla h=\mathcal{N}\cdot B.
	\end{AppB}
	We can then define the following two spaces
	\begin{AppB}
		\label{AppBE2}
		\mathcal{H}_{\operatorname{ex}}(\Omega):=\left\{\nabla f\mid f\in H^1(\Omega)\text{ and }\Delta f=0\text{ in }\Omega\right\},
		\end{AppB}
		\begin{gather}
		\nonumber
		\mathcal{D}:=\left\{\mathcal{N}\cdot \nabla f \mid \nabla f\in \mathcal{H}_{\operatorname{ex}}(\Omega)\right\}.
	\end{gather}
	We have thus shown that $\mathcal{N}\cdot B\in \mathcal{D}$ and it follows from (\cite[Lemma 4.2]{Gerner26KernelImage}) that $$\frac{\operatorname{Id}}{2}+w^{\operatorname{Tr}}_{\Omega}:\mathcal{D}\rightarrow\mathcal{D}$$
	is a well-defined, linear isomorphism. We conclude that $$\left(\left(\frac{\operatorname{Id}}{2}+w^{\operatorname{Tr}}_{\Omega}\right)^{-1}-\operatorname{Id}\right)(\mathcal{N}\cdot B)\in \mathcal{D}$$ and thus the BVP (\ref{S2E3}) admits a solution because the Neumann data satisfies the solvability condition for the Neumann problem.
\end{proof}
We are ready to prove the main result of this manuscript.
\begin{proof}[Proof of \Cref{S2T7}]
	\underline{Step 1: $\operatorname{BS}_{\Sigma}(B\times \mathcal{N}+\nabla f\times \mathcal{N})=B-\operatorname{BS}_P(J)$:} We start by computing $\operatorname{BS}_{\Sigma}(B\times \mathcal{N})$. It follows from the virtual-casing principle \cite{Hanson15},\cite[Lemma 5.5]{G24} that we have the identity
	\begin{gather}
		\nonumber
		\operatorname{BS}_{\Sigma}(B\times \mathcal{N})(x)=B(x)-\operatorname{BS}_{\Omega}(\operatorname{curl}(B))(x)-\nabla_x\frac{\int_{\Omega}B(y)\cdot \frac{x-y}{|x-y|^3}d^3y}{4\pi}\text{ in }\Omega.
	\end{gather}
	We then notice that if we let $h$ denote a solution to the BVP (\ref{AppBE1}), then $\nabla h\in \mathcal{H}_{\operatorname{ex}}(\Omega)$, recall (\ref{AppBE2}), and $B-\nabla h$ is div-free and tangent to the boundary. Hence, $$\int_{\Omega}B(y)\cdot \frac{x-y}{|x-y|^3}d^3y=\int_{\Omega}(\nabla h)(y)\cdot \frac{x-y}{|x-y|^3}d^3y$$
	where we used that $-\nabla_x|x-y|^{-1}=\frac{x-y}{|x-y|^3}$ and that $B-\nabla h$ annihilates all gradient fields. We can then define the following operator
	\begin{AppB}
		\label{AppBE3}
		T:\mathcal{H}_{\operatorname{ex}}(\Omega)\rightarrow\mathcal{H}_{\operatorname{ex}}(\Omega)\text{, }\nabla h\mapsto T(\nabla h)(x):=\nabla_x\frac{\int_{\Omega}\nabla h(y)\cdot \frac{x-y}{|x-y|^3}d^3y}{4\pi}
	\end{AppB}
	which according to \cite[Corollary 3.2]{Gerner26KernelImage} is a well-defined linear operator. We therefore arrive at the identity
	\begin{AppB}
		\label{AppBE4}
		\operatorname{BS}_{\Sigma}(B\times \mathcal{N})=B-\operatorname{BS}_{\Omega}(\operatorname{curl}(B))-T(\nabla h).
	\end{AppB}
	We now apply the virtual-casing principle to $T(\nabla h)$ and find
	$$\operatorname{BS}_{\Sigma}(T(\nabla h)\times \mathcal{N})=T(\nabla h)-T^2(\nabla h)\Leftrightarrow T(\nabla h)=\operatorname{BS}_{\Sigma}(T(\nabla h)\times \mathcal{N})+T^2(\nabla h).$$
	It then follows by induction, that for every fixed $n\in \mathbb{N}$ we have $$T(\nabla h)=\operatorname{BS}_{\Sigma}\left(\sum_{k=1}^nT^k(\nabla h)\times \mathcal{N}\right)+T^{n+1}(\nabla h).$$
	We can insert this into (\ref{AppBE4}) and obtain
	\begin{AppB}
		\label{AppBE5}
		\operatorname{BS}_{\Sigma}\left(\left(B+\sum_{k=1}^nT^k(\nabla h)\right)\times \mathcal{N}\right)=B-\operatorname{BS}_{\Omega}(\operatorname{curl}(B))-T^{n+1}(\nabla h)\text{ in }\Omega\text{ for all }n\in \mathbb{N}.
	\end{AppB}
	It then follows from \cite[Corollary 3.2]{Gerner26KernelImage} that the operator $T$ is a contraction with respect to the $L^2(\Omega)$-norm. That means there exists some $0<\lambda<1$ such that $$\|T(\nabla g)\|_{L^2(\Omega)}\leq \lambda \|\nabla g\|_{L^2(\Omega)}\text{ for all }\nabla g\in \mathcal{H}_{\operatorname{ex}}(\Omega).$$
	This immediately implies $$\|T^{n+1}(\nabla h)\|_{L^2(\Omega)}\leq \lambda^{n+1}\|\nabla h\|_{L^2(\Omega)}\rightarrow 0\text{ as }n\rightarrow\infty.$$
	In addition, $\sum_{k=1}^nT^k(\nabla h)$ is a Cauchy-sequence in $\mathcal{H}_{\operatorname{ex}}(\Omega)$ with respect to $L^2(\Omega)$. This can be seen by computing
	\begin{gather}
		\nonumber
		\left\|\sum_{k=1}^nT^k(\nabla h)-\sum_{k=1}^mT^k(\nabla h)\right\|_{L^2(\Omega)}=\left\|\sum_{k=m+1}^nT^k(\nabla h)\right\|_{L^2(\Omega)}\leq \sum_{k=m+1}^n\|T^k(\nabla h)\|_{L^2(\Omega)}
		\\
		\nonumber
		\leq \|\nabla h\|_{L^2(\Omega)}\sum_{k=m+1}^n\lambda^k\rightarrow 0\text{ as }m,n\rightarrow\infty
	\end{gather}
	where the latter follows because $\sum_{k=1}^n\lambda^k$ converges and is thus a real Cauchy-sequence. We conclude that $$\sum_{k=1}^{n}T^k(\nabla h)\rightarrow \nabla f\in\mathcal{H}_{\operatorname{ex}}(\Omega)$$
	for a suitable $\nabla f\in \mathcal{H}_{\operatorname{ex}}(\Omega)$. We notice that $\|\nabla g\|_{H(\Omega,\operatorname{curl})}=\|\nabla g\|_{L^2(\Omega)}$ for all gradient fields because $\operatorname{curl}(\nabla g)=0$ and therefore, due to the continuity of the twisted tangential trace \cite[Theorem 2.11]{GR86}, we find $$\sum_{k=1}^nT^k(\nabla h)\times \mathcal{N}\rightarrow \nabla f\times \mathcal{N}\text{ in }W^{-\frac{1}{2},2}.$$
	According to \cite[Lemma C.1]{G24} the Biot-Savart operator $\operatorname{BS}_{\Sigma}$ is continuous with respect to the $W^{-\frac{1}{2},2}$-norm on its domain and with respect to the $L^2(\Omega)$ on its range. We may therefore take the limit $n\rightarrow\infty$ in (\ref{AppBE5}) and arrive at the identity
	\begin{AppB}
		\label{AppBE6}
		\operatorname{BS}_{\Sigma}(B\times \mathcal{N}+\nabla f\times \mathcal{N})=B-\operatorname{BS}_{\Omega}(\operatorname{curl}(B))\text{ in }\Omega.
	\end{AppB}
	We will now characterise $f$ as a solution to the Neumann boundary value problem (\ref{S2E3}). It follows immediately, since $\nabla f\in \mathcal{H}_{\operatorname{ex}}(\Omega)$, that $\Delta f=0$ in $\Omega$ and so we are left with characterising its normal trace. To this end we recall that $$\nabla f=\sum_{k=1}^{\infty}T^k(\nabla h)$$ where the convergence takes place in $L^2(\Omega)$. Since $\sum_{k=1}^nT^k(\nabla h)\in \mathcal{H}_{\operatorname{ex}}(\Omega)$ it follows that its $H(\Omega,\operatorname{div})$-norm coincides with its $L^2(\Omega)$-norm and thus, due to the continuity of the normal trace \cite[Theorem 2.5]{GR86}, we see that
	\begin{AppB}
		\label{AppBE7}
		\mathcal{N}\cdot \sum_{k=1}^nT^k(\nabla h)\rightarrow \mathcal{N}\cdot \nabla f\text{ in }W^{-\frac{1}{2},2}.
	\end{AppB}
	To compute the trace we make use of the following fact, \cite[Lemma 4.1]{Gerner26KernelImage},
	$$\mathcal{N}\cdot T(\nabla g)=\left(\frac{\operatorname{Id}}{2}-w^{\operatorname{Tr}}_{\Omega}\right)(\mathcal{N}\cdot \nabla g)\text{ for all }\nabla g\in \mathcal{H}_{\operatorname{ex}}(\Omega).$$ An induction argument shows that this implies $$\mathcal{N}\cdot T^k(\nabla g)=\left(\frac{\operatorname{Id}}{2}-w^{\operatorname{Tr}}_{\Omega}\right)^k(\mathcal{N}\cdot \nabla g)\text{ for all }\nabla g\in \mathcal{H}_{\operatorname{ex}}(\Omega).$$
	Therefore we find $$\mathcal{N}\cdot \sum_{k=1}^nT^k(\nabla h)=\sum_{k=1}^n\left(\frac{\operatorname{Id}}{2}-w^{\operatorname{Tr}}_{\Omega}\right)^k(\mathcal{N}\cdot \nabla h).$$
	It then follows from \cite[Proof of Theorem 2.9]{Gerner26KernelImage} that $$\sum_{k=0}^{n}\left(\frac{\operatorname{Id}}{2}-w^{\operatorname{Tr}}_{\Omega}\right)^k\rightarrow \left(\frac{\operatorname{Id}}{2}+w^{\operatorname{Tr}}_{\Omega}\right)^{-1}$$ in the strong $W^{-\frac{1}{2},2}$-operator norm. In particular, this implies that $$\sum_{k=1}^{n}\left(\frac{\operatorname{Id}}{2}-w^{\operatorname{Tr}}_{\Omega}\right)^k(\mathcal{N}\cdot \nabla h)\rightarrow\left(\left(\frac{\operatorname{Id}}{2}+w^{\operatorname{Tr}}_{\Omega}\right)^{-1}-\operatorname{Id}\right)(\mathcal{N}\cdot \nabla h)\text{ in }W^{-\frac{1}{2},2}.$$ We recall, (\ref{AppBE1}), that $\mathcal{N}\cdot \nabla h=\mathcal{N}\cdot B$ and it then follows from (\ref{AppBE7}) that
	\begin{gather}
		\nonumber
		\mathcal{N}\cdot \nabla f=\left(\left(\frac{\operatorname{Id}}{2}+w^{\operatorname{Tr}}_{\Omega}\right)^{-1}-\operatorname{Id}\right)(\mathcal{N}\cdot B)
	\end{gather}
	as claimed. Finally, in our application we know that $B$ is a vacuum field outside of the plasma region $P$, which implies $\operatorname{BS}_{\Omega}(\operatorname{curl}(B))=\operatorname{BS}_P(J)$ on all of $\mathbb{R}^3$ where $J=\operatorname{curl}(B)$. So (\ref{AppBE6}) becomes
	\begin{gather}
		\nonumber
		\operatorname{BS}_{\Sigma}(B\times \mathcal{N}+\nabla f\times \mathcal{N})=B-\operatorname{BS}_P(J)\text{ in }\Omega
	\end{gather}
	as was to be shown.
	\newline
	\newline
	\underline{Step 2: Concluding the proof:} We recall that according to \Cref{S2T5}, if $\widetilde{\Gamma}_p\in \mathcal{H}_N(\widetilde{\Omega})$ is the unique element with unit poloidal circulation along a fixed poloidal curve, then $\operatorname{BS}_{\Sigma}(\alpha \widetilde{\Gamma}_p\times \mathcal{N})=0$ in $\Omega$ for all $\alpha\in \mathbb{R}$. Consequently we are left with proving that the choice $\alpha=\int_{\sigma_p}B$ leads to the most poloidal field lines of the current $$j=B\times \mathcal{N}+\nabla f\times \mathcal{N}+\alpha \widetilde{\Gamma}_p\times \mathcal{N}.$$
	To quantify this property we make use of the notions of average poloidal and toroidal windings, cf. \cite[Definition 2.11 \& Lemma 2.12]{Ger26SurfaceHelicity}, which can be defined as follows: Fix a toroidal loop $\sigma_t$ on $\Sigma$ and let $\sigma_p$ be our fixed poloidal loop. Then there exists a unique basis $\gamma_p,\gamma_t$ of the space $$\mathcal{H}(\Sigma):=\{\gamma \in \mathcal{V}(\Sigma)\mid \operatorname{div}_{\Sigma}(\gamma)=0=\operatorname{curl}_{\Sigma}(\gamma)\},$$
	where $\mathcal{V}(\Sigma)$ denotes the space of smooth tangent vector fields on $\Sigma$, satisfying
	\begin{AppB}
		\label{AppBE8}
		\int_{\sigma_p}\gamma_p=1=\int_{\sigma_t}\gamma_t\text{ and }\int_{\sigma_p}\gamma_t=0=\int_{\sigma_t}\gamma_p.
	\end{AppB}
	The average poloidal windings $\overline{P}(j)$ and average toroidal windings $\overline{Q}(j)$ of a div-free current distribution $j$ on $\Sigma$ can then be defined by
	\begin{gather}
		\nonumber
		\overline{P}(j)=\frac{\int_{\Sigma}\gamma_p\cdot jdS}{|\Sigma|}\text{, }\overline{Q}(j)=\frac{\int_{\Sigma}\gamma_t\cdot jdS}{|\Sigma|}.
	\end{gather}
	Clearly $\overline{P}$ and $\overline{Q}$ are linear and it also follows immediately that $\overline{Q}(\nabla f\times \mathcal{N})=0=\overline{P}(\nabla f\times \mathcal{N})$ since the space of harmonic fields $\mathcal{H}(\Sigma)$ is $L^2(\Sigma)$-orthogonal to the co-exact vector fields. It follows further from the fact that $B$ is curl-free outside of $P$, that the part $B^{\parallel}$ of $B|_{\Sigma}$ which is tangent of $\Sigma$ satisfies $\operatorname{curl}_{\Sigma}(B^{\parallel})=0$. It follows from the Hodge decomposition theorem that we can write $$B^{\parallel}=\nabla_{\Sigma}g+a\gamma_p+b\gamma_t$$ for a suitable smooth function $g$ on $\Sigma$ and $a,b\in \mathbb{R}$. By the defining properties of $\gamma_p$ and $\gamma_t$ we find $$a=\int_{\sigma_p}B\text{ and }b=\int_{\sigma_t}B$$
	where we used that $\sigma_p$ and $\sigma_t$ are tangent to $\Sigma$. We then compute
	\begin{gather}
		\nonumber
		|\Sigma|\overline{Q}(B\times \mathcal{N})=|\Sigma|\overline{Q}(B^{\parallel}\times \mathcal{N})=\left(\int_{\sigma_p}B\right)\int_{\Sigma}(\gamma_p\times \mathcal{N})\cdot \gamma_tdS
	\end{gather}
	where we used that $\gamma_t$ is $L^2$-orthogonal to the co-exact fields and that $\gamma_t\times \mathcal{N}$ and $\gamma_t$ are everywhere pointwise orthogonal. We arrive at
	\begin{AppB}
		\label{AppBE9}
		\overline{Q}(B\times \mathcal{N}\times \nabla f\times \mathcal{N})=\left(\int_{\sigma_p}B\right)\frac{\int_{\Sigma}(\gamma_p\times \mathcal{N})\cdot \gamma_tdS}{|\Sigma|}.
	\end{AppB}
	Now in a similar fashion, since $\operatorname{curl}(\widetilde{\Gamma}_p)=0$, we find $$\widetilde{\Gamma}_p|_{\Sigma}=\nabla_{\Sigma}\widetilde{g}+\gamma_p$$
	for a suitable smooth function $\widetilde{g}$ on $\Sigma$ and where we used that $\widetilde{\Gamma}_p$ is tangent to $\Sigma$, that $\int_{\sigma_t}\widetilde{\Gamma}_p=0$ according to the representation formula (\ref{AppAE2}) and where we used that $\int_{\sigma_p}\widetilde{\Gamma}_p=1$ by definition of $\widetilde{\Gamma}_p$. We therefore compute
	$$\overline{Q}(\widetilde{\Gamma}_p\times \mathcal{N})=\frac{\int_{\Sigma}(\gamma_p\times \mathcal{N})\cdot \gamma_tdS}{|\Sigma|}$$ and hence (\ref{AppBE9}) becomes
	\begin{gather}
		\nonumber
		\overline{Q}\left(B\times \mathcal{N}+\nabla f\times \mathcal{N}-\left(\int_{\sigma_p}B\right)\widetilde{\Gamma}_p\times \mathcal{N}\right)=0.
	\end{gather}
	Consequently, the choice $\alpha=\int_{\sigma_p}B$ leads to the current of most poloidal field lines.
\end{proof}
\begin{rem}
	\label{AppBR1}
	\begin{enumerate}
		\item Since $B$ is curl-free within the region between the CWS $\Sigma$ and the plasma boundary $\partial P$, it follows that if $\widetilde{\sigma}_p$ is any poloidal loop on $\partial P$ (whose orientation is compatible with that of the poloidal loop $\sigma_p$ on $\Sigma$) then $\int_{\sigma_p}B=\int_{\widetilde{\sigma}_p}B$ and therefore the value of $\alpha$ in \Cref{S2T7} can be computed directly from the knowledge of $B$ on the plasma boundary alone.
		\item We have shown that $\overline{Q}\left(B\times \mathcal{N}+\nabla f\times \mathcal{N}-\left(\int_{\sigma_p}B\right)\widetilde{\Gamma}_p\times \mathcal{N}\right)=0$. One can similarly verify that $$\overline{P}\left(B\times \mathcal{N}+\nabla f\times \mathcal{N}-\left(\int_{\sigma_p}B\right)\widetilde{\Gamma}_p\times \mathcal{N}\right)=\left(\int_{\sigma_t}B\right)\overline{P}(\gamma_t\times \mathcal{N})$$
		and it is always the case that $\overline{P}(\gamma_t\times \mathcal{N})\neq 0$. Hence, as long as $\int_{\sigma_t}B\neq 0$, the corresponding current $B\times \mathcal{N}+\nabla f\times \mathcal{N}-\left(\int_{\sigma_p}B\right)\widetilde{\Gamma}_p\times \mathcal{N}$ will have a poloidal component.
	\end{enumerate} 
\end{rem}
\section{Conditions for poloidal current field lines}
\label{AppC}
We start by considering a very general situation.
\begin{prop}
\label{AppCP1}
Let $\Omega\subset\mathbb{R}^3$ be a (possibly unbounded) smooth domain such that $\partial \Omega$ is a $2$-torus. Let $X$ be a smooth vector field on $\overline{\Omega}$ such that $X\parallel \partial \Omega$ and such that $X|_{\partial\Omega}$ does not vanish identically. Then the following two sets of conditions are equivalent:
\begin{enumerate}
	\item $X|_{\partial\Omega}$ is no-where vanishing.
	\item There exist smooth vector fields $Y,Z$ on $\Sigma$ such that the following holds
	\begin{enumerate}
		\item $\operatorname{div}(X)-\mathcal{N}\cdot \nabla_{\mathcal{N}}X=Y\cdot X\text{ on }\partial\Omega$,
		\item $\mathcal{N}\cdot \operatorname{curl}(X)=Z\cdot X\text{ on }\partial\Omega$
	\end{enumerate}
	where $\operatorname{div}$ and $\operatorname{curl}$ denote the standard div- and curl operators in $3$-space and $\mathcal{N}\cdot \nabla_{\mathcal{N}}X=((DX)\cdot \mathcal{N})\cdot\mathcal{N}$ where $DX$ denotes the Jacobian of $X$ and $\mathcal{N}$ is the outward unit normal on $\partial\Omega$.
\end{enumerate}
\end{prop}
\begin{proof}[Proof of \cref{AppCP1}]
	\underline{(i) $\Rightarrow$ (ii):} Suppose that $X|_{\partial\Omega}$ does not vanish. Then there is some neighbourhood around $\partial \Omega$ on which the function $f:=|X|^{-2}$ is well-defined and smooth. We then compute
	\begin{gather}
		\nonumber
		\operatorname{div}(X)=\operatorname{div}(|X|^2fX)=|X|^2\operatorname{div}(fX)+f\nabla (|X|^2)\cdot X\text{, }
		\\
		\nonumber
		\mathcal{N}\cdot \nabla_{\mathcal{N}}X=\mathcal{N}\cdot \nabla_{\mathcal{N}}(|X|^2fX)=|X|^2\mathcal{N}\cdot \nabla_{\mathcal{N}}(fX)+(\mathcal{N}\cdot \nabla |X|^2)f\mathcal{N}\cdot X=|X|^2\mathcal{N}\cdot \nabla_{\mathcal{N}}(fX)
	\end{gather}
	where we used that $\mathcal{N}\cdot X=0$ in the last step. Setting $$Y:=X(\operatorname{div}(fX)-\mathcal{N}\cdot \nabla_{\mathcal{N}}(fX))+f\nabla |X|^2$$ we see that $Y$ is smooth and $$\operatorname{div}(X)-\mathcal{N}\cdot \nabla_{\mathcal{N}}X=X\cdot Y$$
	as required.
	
	Similarly one computes
	\begin{gather}
		\nonumber
		\operatorname{curl}(X)=\operatorname{curl}(|X|^2fX)=|X|^2\operatorname{curl}(fX)+f\nabla (|X|^2)\times X
		\\
		\nonumber
		\Rightarrow \mathcal{N}\cdot \operatorname{curl}(X)=|X|^2\mathcal{N}\cdot \operatorname{curl}(fX)+f\mathcal{N}\cdot \left(\nabla |X|^2\times X\right)=|X|^2\mathcal{N}\cdot \operatorname{curl}(fX)+f\left(\mathcal{N}\times \nabla |X|^2\right)\cdot X
	\end{gather}
	where we used the scalar triple product rule in the last step. Thus, $$Z:=\left(\mathcal{N}\cdot \operatorname{curl}(fX)\right)X+f\left(\mathcal{N}\times \nabla |X|^2\right)$$
	is smooth and satisfies $\mathcal{N}\cdot \operatorname{curl}(X)=Z\cdot X$ and the implication is shown.
	\newline
	\newline
	\underline{(ii) $\Rightarrow$ (i):} Let us set $v:=X|_{\partial\Omega}$. The imposed conditions on $X$ imply
	\begin{AppC}
		\label{AppCE1}
		\operatorname{curl}_{\Sigma}(v)=Z\cdot v\text{ and }\operatorname{div}_{\Sigma}(v)=Y\cdot v\text{ on }\partial\Omega.
	\end{AppC}
	Since $Y$ and $Z$ are smooth it follows from \cite{AKS62} that either $v$ is identically zero on $\partial\Omega$ or otherwise only has zeros of finite order. It then follows from \cite[Main theorem]{B99} that if $v$ is not identically zero, then its zero set is discrete. Then the proof of \cite[Lemma 5.4]{GerNB26ElectricCurrentArXiv} applies mutatis mutandis and we conclude that if $p\in \partial\Omega$ is a zero of $v$, then the index of $v$ at $p$ satisfies $\operatorname{ind}_p(v)=-(m(p)+1)$ where $m(p)$ denotes the multiplicity of the zero. In particular $\operatorname{ind}_p(v)\leq -1$ at every zero $p$ of $v$. Then the Poincar\'{e}-Hopf theorem, \cite[Section 6]{M65}, implies that $v$ cannot have any zeros because the Euler characteristic of $\partial\Omega$ is zero.
\end{proof}
\begin{rem}
	\label{AppCR2}
	A related result has been proven in \cite[Theorem 8]{PPS22} which implies that if $X$ is a smooth vector field on $\overline{\Omega}$, tangent to $\partial\Omega$ (as usual $\partial\Omega$ is assumed to be a torus), and there exists a smooth function $u$ on $\partial\Omega$ such that $\operatorname{div}(X)-\mathcal{N}\cdot \nabla_{\mathcal{N}}X=X\cdot \nabla u$ and $\mathcal{N}\cdot \operatorname{curl}(X)=0$, then either $X$ vanishes identically on $\partial\Omega$ or otherwise is no-where vanishing and semi-rectifiable. The latter means that there is a smooth, positive function $f>0$ on $\partial\Omega$ and a global coordinate system on $\partial\Omega$ such that $fX$ is represented by $a\partial_{\phi}+b\partial_{\theta}$ in this coordinate system for suitable $a,b\in \mathbb{R}$. Thus, while the assumptions in \cite{PPS22} are more restrictive, they also provide a stronger conclusion.
\end{rem}
We obtain the following corollary of relevance in plasma confinement.
\begin{cor}
	\label{AppCC3}
	Let $P\subset\mathbb{R}^3$ be a bounded smooth solid torus and suppose that $\partial P$ is analytic. Let $B$ be a smooth vector field on $\overline{P}$ solving the plasma equilibrium equation (\ref{S1EExtra}) and assume that $B|_{\partial P}$ is an analytic vector field on $\partial P$ with $\int_{\sigma_t}B\neq 0$ for some toroidal loop $\sigma_t$ on $\partial P$. Let $U$ denote the open neighbourhood of $\overline{P}$ on which $B$ admits a (unique) compatible vacuum extension, cf. \Cref{S2T1}, and let $\Sigma\subset U$ be a smooth toroidal surface such that $\overline{P}\subset \Omega$ and such that $\sigma_t$ is not contractible within $\Omega$, where $\Omega$ denotes the finite region bounded by $\Sigma$. Assume further that $B\parallel \Sigma$, where we set $B$ to its vacuum extension outside of $P$. Then the following two conditions are equivalent for the current $j_p$ as defined in (\ref{S4E2}):
	\begin{enumerate}
		\item All field lines of $j_p$ are closed poloidal loops on $\Sigma$.
		\item There exists a smooth vector field $Y$ on $\Sigma$, such that $\mathcal{N}\cdot \nabla_{\mathcal{N}}B_{\operatorname{eff}}=Y\cdot B_{\operatorname{eff}}$ on $\partial\Omega$, where $B_{\operatorname{eff}}=B-\left(\int_{\sigma_p}B\right)\widetilde{\Gamma}_p$ and $\sigma_p$, $\widetilde{\Gamma}_p$ are as in \Cref{S2P4} and $\mathcal{N}$ is the outward unit normal on $\Sigma$ with respect to $U\setminus \overline{\Omega}$.
	\end{enumerate}
\end{cor}
\begin{proof}[Proof of \Cref{AppCC3}]
	\underline{(i) $\Rightarrow$ (ii):} (i) implies that $j_p$ is no-where vanishing, since otherwise it would admit a field line (a stationary point) which is not a poloidal loop. Since $j_p=B_{\operatorname{eff}}\times\mathcal{N}$ and $B_{\operatorname{eff}}\parallel \Sigma$ because $B\parallel \Sigma$ by assumption, we find that $B_{\operatorname{eff}}$ is no-where vanishing on $\Sigma$. We observe then that $B$, as well as $\widetilde{\Gamma}_p$ are defined on $U\setminus \Omega$ and are curl- and div-free on this set. It then follows from \Cref{AppCP1} that condition (ii,a) of \Cref{AppCP1} holds. Since $\operatorname{div}(B_{\operatorname{eff}})=0$ we conclude that condition (ii) of \Cref{AppCC3} must hold.
	\newline
	\newline
	\underline{(ii) $\Rightarrow$ (i):} By assumption $\int_{\sigma_t}B\neq 0$ for a toroidal loop $\sigma_t$ on $\partial P$. Since $B$ is curl-free within $U\setminus \overline{P}$ and $\sigma_t$ is not contractible within $\Omega$, there exists a toroidal loop $\widetilde{\sigma}_t$ on $\partial\Omega=\Sigma$ such that $\int_{\widetilde{\sigma}_t}B=\int_{\sigma_t}B\neq 0$. At the same time we have $\int_{\widetilde{\sigma}_t}\widetilde{\Gamma}_p=0$, which follows from the representation formula (\ref{AppAE2}). We conclude that $\int_{\widetilde{\sigma}_t}B_{\operatorname{eff}}\neq 0$ and thus $B_{\operatorname{eff}}$ is not identically zero on $\Sigma$. Since $B$ and $\widetilde{\Gamma}_p$ are div- and curl-free on $U\setminus \overline{\Omega}$ it follows from assumption (ii) that $B_{\operatorname{eff}}$ satisfies condition (ii) of \Cref{AppCP1}. Therefore \Cref{AppCP1} implies that $B_{\operatorname{eff}}$ is no-where vanishing on $\Sigma$ and thus $j_p$ also does not vanish on $\Sigma$. We have further shown at the end of the proof of \Cref{S2T7} and \Cref{AppBR1} that $\overline{Q}(j_p)=0$ and $\overline{P}(j_p)\neq 0$ (because $\int_{\sigma_t}B\neq 0$) where $\overline{Q}(j_p)$ and $\overline{P}(j_p)$ denote the average toroidal and poloidal windings respectively. It then follows from \cite[Lemma B.1]{Ger25NonoptimalHelicityArXiv} that all field lines of $j_p$ are poloidal loops on $\Sigma$.
\end{proof}
\bibliographystyle{plain}
\bibliography{mybibfileNOHYPERLINK}
\footnotesize
\end{document}